\documentclass[a4paper,11pt]{article}
 \usepackage[T1]{fontenc}
 \usepackage[normalem]{ulem}
 \usepackage{verbatim}
 \usepackage{graphicx}
 \usepackage{amsmath}
 \usepackage{amssymb}
\usepackage{geometry}
\usepackage{amsthm}
\usepackage[utf8]{inputenc}

\newcommand{\R}{\mathbb{R}}
\newcommand{\N}{\mathbb{N}}
\newcommand{\Z}{\mathbb{Z}}

\newcommand{\cC}{\mathcal{C}}

\newcommand{\cQ}{\mathcal{Q}}
\newcommand{\cR}{\mathcal{R}}

\newcommand{\tr}{{\rm Tr}}

\newcommand{\dps}{\displaystyle }

\newtheorem{thm}{Theorem}[section]  
\newtheorem{prop}{Proposition}[section]
\newtheorem{lem}{Lemma}[section]
\newtheorem{remark}{Remark}[section]

\title{Local defects are always neutral in the 
\\ Thomas-Fermi-von Weiszäcker theory of crystals}
\author{Eric Cancès and Virginie Ehrlacher \\ \\
{\small Universit\'e Paris-Est, CERMICS, Project-team Micmac, INRIA-Ecole des Ponts,} \\
{\small 6 \& 8 avenue Blaise Pascal, 77455 Marne-la-Vall\'ee Cedex 2, France} 
}

\begin{document}

\maketitle

\begin{abstract}
{The aim of this article is to propose a mathematical model describing the electronic structure of crystals with local defects in the framework of the Thomas-Fermi-von Weizsäcker (TFW) theory. The approach follows the same lines as that used in  {\it E. Canc\`es, A. Deleurence and M. Lewin, Commun. Math. Phys., 281 (2008), pp. 129--177} for the reduced Hartree-Fock model, and is based on thermodynamic limit arguments. We prove in particular that it is not possible to model charged defects within the TFW theory of crystals. We finally derive some additional properties of the TFW ground state electronic density of a crystal with a local defect, in the special case when the host crystal is modelled by a homogeneous medium.
}
\end{abstract}


\section{Introduction}

The modelling and simulation of the electronic structure of crystals is a prominent topic in solid-state physics, materials science and 
nano-electronics~\cite{Kittel,Pisani,Stoneham}. Besides its importance for the applications, it is an interesting playground for mathematicians 
for it gives rise to many interesting mathematical and numerical questions.

There are two reasons why the modelling and simulation of the electronic structure of crystals is a difficult task. First, the number of particles in a crystal 
is infinite, and second, the Coulomb interaction is long-range. Of course, a real crystal contains a finite number of electrons and nuclei, but in order to 
understand and compute the macroscopic properties of a crystal from first principles, it is in fact easier, or at least not more complicated, to consider that 
we are dealing with an infinite system. 

\medskip

The first mathematical studies of the electronic structure of crystals were concerned with the so-called thermodynamic limit problem for perfect crystals. 
This problem can be stated as follows. Starting from a given electronic structure model for finite molecular systems, find out an electronic structure model 
for perfect crystals, such that when a cluster grows and ``converges'' (in some sense, see~\cite{Cat98}) to some $\cR$-periodic perfect crystal, the ground state 
electronic density of the cluster converges to the $\cR$-periodic ground state electronic density of the perfect crystal.

For Thomas-Fermi like (orbital-free) models, it is not difficult to guess what should be the corresponding models for perfect crystals. On the other hand, solving 
the thermodynamic limit problem, that is proving the convergence property discussed above, is a much more difficult task. This program was carried out for the Thomas-Fermi (TF) model in \cite{LiebSimon} and for the 
Thomas-Fermi-von Weizsäcker (TFW) model in \cite{Cat98}. Note that these two models are strictly convex in the density, and that the uniqueness of the ground 
state density is an essential ingredient of the proof. The thermodynamic limit problem for perfect crystals remains open for the Thomas-Fermi-Dirac-von 
Weisäcker model, and more generally for  nonconvex orbital-free models.

The case of Hartree-Fock and Kohn-Sham like models is more difficult. In these models, the electronic state is described in terms of electronic 
density matrices. For a finite system, the ground state density matrix is a non-negative trace-class self-adjoint operator, with trace 
$N$, the number of electrons in the system. For infinite systems, the ground state density matrix is no longer trace-class, which significantly complicates 
the mathematical arguments. Yet, perfect crystals being periodic, it is possible to make use of Bloch-Floquet theory and guess the structure of the 
periodic Hartree-Fock and Kohn-Sham models. These models are widely used in solid-state physics and materials science. Here also, the thermodynamic limit problem 
seems out of reach with state-of-the-art mathematical tools, except in the special case of the restricted Hartree-Fock (rHF) model, also called the Hartree model 
in the physics literature. Thoroughly using the strict convexity of the rHF energy functional with respect to the electronic density, Catto, Le Bris and Lions were 
able to solve the thermodynamic limit problem for the rHF model~\cite{CatBriLio-01}. 

Very little is known about the modelling of perfect crystals within the framework of the $N$-body Schrödinger model. To the best of our knowledge, the 
only available results \cite{Fefferman,HaiLewSol_thermo-08} state that the energy per unit volume is well defined in the thermodynamic limit. So far, the Schrödinger model 
for periodic crystals is still an unknown mathematical object.

\medskip

The mathematical analysis of the electronic structure of crystals with defects has been initiated in \cite{CDL1} for the rHF model. This work is based on a 
formally simple idea, whose rigorous implementation however requires some effort. This idea is very similar to that used in \cite{ChaIra-89,HaiLewSer-08,HaiLewSerSol-07} 
to properly define a no-photon quantum electrodynamical (QED) model for atoms and molecules. Loosely speaking, it consists in considering the defect (the atom or the 
molecule in QED) as a quasiparticle embedded in a well-characterized background (a perfect crystal in our case, the polarized vacuum in QED), and to build a variational 
model allowing to compute the ground state of the quasiparticle. 

In \cite{CDL1}, such a variational model is obtained by passing to the thermodynamic limit in the difference between the ground state density matrices obtained 
respectively with and without the defect. In order to avoid additional technical difficulties, the thermodynamic limit argument in~\cite{CDL1} is not carried out on 
clusters (as in \cite{Cat98,CatBriLio-01}), but on supercells of increasing sizes. Recall that the supercell model is the current state-of-the-art method to compute 
the electronic structure of a crystal with a local defect. In this approach, the defect and as many atoms of the host crystal as the available computer resources can 
accomodate, are put in a large, usually cubic, box, called the supercell, and Born-von-Karman periodic boundary conditions are imposed to the single particle orbitals 
(and consequently to the electronic density). The limitations of the supercell methods are well-known: first, it gives rise to spurious interactions between the defect 
and its periodic images, and second, it requires that the total charge contained in the supercell is neutral (otherwise, the energy per unit volume would be infinite). In 
the case of charged defects, the extra amount of charge must be compensated in one way or another, for instance by adding to the total physical charge distribution of the 
system a uniformly charged background (called a jellium). It is well-known that this procedure generates unphysical screening effects. Other charge compensation methods have 
been proposed, but none of them is completely satisfactory. Note that the above mentioned sources of error vanish in the thermodynamic limit, when the size of the supercell 
goes to infinity: both the interaction between a defect and its periodic images and the density of the jellium go to zero in the thermodynamic limit.

The variational model for the defect, considered as a quasiparticle, obtained in \cite{CDL1} has a quite unusual mathematical structure. The rHF ground state density 
matrix of the crystal in the presence of the defect can be written as
$$
\gamma = \gamma^0_{\rm per} + Q
$$
where $\gamma^0_{\rm per}$ is the density matrix of the host perfect crystal (an orthogonal projector on $L^2(\R^3)$ with infinite rank which commutes with the translations 
of the lattice) and $Q$ a self-adjoint Hilbert-Schmidt operator on $L^2(\R^3)$. Although $Q$ is not trace-class in general~\cite{CL}, it is possible to give a sense 
to its generalized trace
$$
\tr_0(Q) := \tr(Q^{++})+\tr(Q^{--}) \quad \mbox{where} \quad 
Q^{++} := (1-\gamma^0_{\rm per}) Q (1-\gamma^0_{\rm per}) \mbox{ and } Q^{--} := \gamma^0_{\rm per}Q\gamma^0_{\rm per}
$$
(as $\gamma^0_{\rm per}$ is an orthogonal projector, $\tr=\tr_0$ on the space of the trace-class operators on $L^2(\R^3)$), as well as to its density $\rho_Q$. 
The latter is defined in a weak sense
$$
\forall W \in C^\infty_{\rm c}(\R^3), \quad \tr_0(QW)=\int_{\R^3} \rho_Q W.
$$
The function $\rho_Q$ is not in $L^1(\R^3)$ in general, but only in $L^2(\R^3) \cap \cC$, where $\cC$ is the Coulomb space defined by~(\ref{eq:defC}). 
An important consequence of these results is that
\begin{itemize}
\item in general, the electronic charge of the defect can be defined neither as $\tr(Q)$ nor as $\int_{\R^3} \rho$
\item it may happen that $\rho_Q \in L^1(\R^3)$ but $\tr_0(Q) \neq \int_{\R^3} \rho_Q$ (while we would have 
$\rho_Q \in L^1(\R^3)$ and $\tr_0(Q)=\tr(Q)=\int_{\R^3}\rho_Q$ if $Q$ were a trace-class operator). In this case, $\tr_0(Q)$ and $\int_{\R^3} \rho_Q$ can be 
interpreted respectively as the {\em bare} and {\em renormalized} electronic charges of the defect~\cite{CL}.  
\end{itemize}
The reason why, in general, $Q$ is not trace-class and $\rho_Q$ is not an integrable function, is a consequence of both the infinite number of particles 
and the long-range of the Coulomb interaction.

Note that, still in the rHF setting, the dynamical version of this variational model is nothing but the random phase approximation (RPA), widely used in 
solid-state physics. The well-posedness of the nonlinear RPA dynamics, as well as of each term of the Dyson expansion with respect to the external potential, 
is proved in~\cite{CancesStoltz}.

\medskip

As far as we know, the mathematical study of the electronic structure of crystals with local defects has not been completed for the Thomas-Fermi-von Weizsäcker 
model~\cite{Lieb}. This is the purpose of the present work. The article is organized as follows. In Section~\ref{sec:TFW}, we present the periodic TFW model used in 
condensed phase calculations. After recalling the mathematical structure of the TFW model for perfect crystals (Section~\ref{sec:perfect}), we propose a variational TFW model for crystals with local defects (Section~\ref{sec:TFWdef}). We prove that this model is well-posed and 
that the nuclear charge of the defect is fully screened, in a sense that will be precisely defined. In Section~\ref{sec:TL}, we provide a mathematical justification of 
the model introduced in Section~\ref{sec:defect} based on bulk limit arguments. In Section~\ref{sec:HEG}, we focus on the special case when the host crystal is a homogeneous 
medium, that is when both the nuclear and electronic densities of the host crystal are uniform (and opposite one another to prevent Coulomb blow-up). The technical parts of the 
proofs are gathered in Section~\ref{sec:proofs}. 

\medskip

Note that the screening effect has already been studied in the context of the Thomas-Fermi model in~\cite{LiebSimon}, in the case when the host crystal is a homogeneous medium.

\section{The periodic Thomas-Fermi-von Weiszäcker model}
\label{sec:TFW}

In this section, we describe the Thomas-Fermi-von Weiszäcker (TFW) model with Born-von Karman periodic boundary conditions, used to perform calculations in 
the condensed phase. In Section~\ref{sec:TL}, we will use this periodic model to pass to the thermodynamic limit and construct a rigorously founded TFW model for crystals 
with local defects. 
 
\medskip

Let $\cR$ be a periodic lattice of $\R^3$, $\cR^*$ the associated reciprocal lattice, and $\Gamma$ the simulation cell. If for instance $\cR = a\Z^3$ 
(cubic lattice of size $a$), then $\cR^\ast=\frac{2\pi}a\Z^3$ and possible choices for $\Gamma$ are $\Gamma=(0,a]^3$ or $\Gamma=(-\frac a2,\frac a2]^3$.  Let also $\Gamma^*$ be the first Brillouin zone of the lattice $\cR$ (or in other words, the Wigner-Seitz cell of the reciprocal lattice $\cR^*$).

\medskip

We introduce the usual $\cR$-periodic $L^p$ spaces defined by
$$
L^p_{\rm per}(\Gamma) := \left\{ v \in L^p_{\rm loc}(\mathbb{R}^3) \; |\; v \; \mathcal{R}\mbox{-periodic} \right\},
$$
and endow them with the norms
$$
\|v\|_{L^p_{\rm per}(\Gamma)} := \left( \int_\Gamma |v|^p \right)^{1/p} \quad \mbox{for } 1 \le p < \infty \quad \mbox{and} \quad \|v\|_{L^\infty_{\rm per}(\Gamma)} := \mbox{ess-sup}|v|.
$$
In particular,
$$ 
\|v\|_{L^2_{\rm per}(\Gamma)}=(v,v)_{L^2_{\rm per}(\Gamma)}^{1/2} \quad \mbox{where} \quad 
\left( v, w \right)_{L^2_{\rm per}(\Gamma)} := \int_{\Gamma} v w.
$$
Any function $v\in L^2_{\rm per}(\Gamma)$ can be expanded in Fourier modes as 
$$ 
v(x) = \sum_{k\in\mathcal{R}^*} c_k(v) \frac{e^{ik\cdot x}}{|\Gamma|^{1/2}} \quad \mbox{where} \quad c_k(v) = \frac{1}{|\Gamma|^{1/2}}\int_{\Gamma} v(x) e^{-ik\cdot x}\,dx.
$$
The convergence of the above series holds in $L^2_{\rm per}(\Gamma, \mathbb{C})$, the space of locally square integrable $\cR$-periodic $\mathbb{C}$-valued functions.

For each $s \in \R$, the $\cR$-periodic Sobolev space of index $s$ is defined as
$$
H^s_{\rm per}(\Gamma) := \left\{ v(x) = \sum_{k \in \cR^\ast} c_k(v) \frac{e^{ik\cdot x}}{|\Gamma|^{1/2}} \; |\; \sum_{k \in \cR^\ast} (1+|k|^2)^s|c_k(v)|^2 < \infty, \; \forall k \in \cR^\ast, \; c_{-k}=\overline{c_k} \right\},
$$
and endowed with the inner product
$$
(v,w)_{H^s_{\rm per}(\Gamma)} := \sum_{k\in\mathcal{R}^*}(1+|k|^2)^s \overline{c_k(v)}c_k(w).
$$
The condition $\forall k \in \cR^\ast, \; c_{-k}=\overline{c_k}$ implies that the functions of $H^s_{\rm per}(\Gamma)$ are real-valued. Recall that $H^0_{\rm per}(\Gamma)=L^2_{\rm per}(\Gamma)$, $(\cdot,\cdot)_{H^0_{\rm per}(\Gamma)}=(\cdot,\cdot)_{L^2_{\rm per}(\Gamma)}$, 
$$
H^1_{\rm per}(\Gamma)=\left\{ v \in L^2_{\rm per}(\Gamma) \; | \, \nabla v \in 
\left(L^2_{\rm per}(\Gamma)\right)^3 \right\}, \quad
(v,w)_{H^1_{\rm per}(\Gamma)}=\int_\Gamma vw+ \int_\Gamma \nabla v \cdot \nabla w,
$$
and $(H^{-\sigma}_{\rm per}(\Gamma))'=H^{\sigma}_{\rm per}(\Gamma)$.

\medskip

We also introduce the $\mathcal{R}$-periodic Coulomb kernel $G_\cR$ defined as the unique function of $L^2_{\rm per}(\Gamma)$ solution of the elliptic problem
$$\left\{
\begin{array}{l}
\dps -\Delta G_\cR = 4\pi\left( \sum_{k\in \mathcal{R}} \delta_k - |\Gamma|^{-1} \right)  \\
\dps G_\cR \; \mathcal{R}\mbox{-periodic}, \; \min_{\mathbb{R}^3} G_\cR = 0.\\
\end{array}
\right.
$$
It is easy to check that
$$
G_\cR(x) = \frac{1}{|\Gamma|} \int_\Gamma G_\cR + \sum_{k\in \mathcal{R}^* \setminus \{0\}} \frac{4\pi}{|k|^2}\frac{e^{ik\cdot x}}{|\Gamma|}.
$$
The $\cR$-periodic Coulomb energy is then defined for all $f$ and $g$ in $L^2_{\rm per}(\Gamma)$ by 
\begin{eqnarray*}
D_\cR(f,g) &=& \int_{\Gamma} \int_{\Gamma} G_\cR(x-y)f(x)g(y)\,dx\,dy \\ &=&
\left(\int_\Gamma G_\cR\right)\overline{c_0(f)} c_0(g) + \sum_{k\in \mathcal{R}^* \setminus \{0\}} \frac{4\pi}{|k|^2} \overline{c_k(f)}{c_k(g)} 
\\
&=& \int_{\Gamma} (G_\cR \star_\cR f)(y)g(y) \, dy = \int_{\Gamma} (G_\cR \star_\cR g)(x) f(x) \, dx,
\end{eqnarray*}
where $\star_\cR$ denotes the $\cR$-periodic convolution product: 
$$
\forall (f,g) \in L^2_{\rm per}(\Gamma) \times L^2_{\rm per}(\Gamma), \quad
(f \star_\cR g)(x) = \int_{\Gamma} f(x-y) g(y)\,dy = \int_{\Gamma } f(y) g(x-y) \,dy.
$$

\medskip

Let $\rho^{\rm nuc}$ be a function of $L^2_{\rm per}(\Gamma)$ modelling a $\cR$-periodic nuclear charge distribution (or the effective charge distribution of a pseudopotential describing a $\cR$-periodic distribution of nuclei and core electrons). The corresponding $\cR$-periodic TFW energy functional is defined on $H^1_{\rm per}(\Gamma)$ and reads
\begin{equation} \label{eq:TFWenergyper}
E_\cR^{\rm TFW}(\rho^{\rm nuc}, v) = C_{\rm W} \int_{\Gamma} |\nabla v|^2 + C_{\rm TF} \int_{\Gamma} |v|^{10/3} + \frac{1}{2} D_\cR(\rho^{\rm nuc}-v^2, \rho^{\rm nuc}-v^2),
\end{equation}
where 
$$
C_{\rm TF}=\frac{10}{3} (3\pi^2)^{2/3} \quad \mbox{(Thomas-Fermi constant)} \quad \mbox{ and } \quad 
C_{\rm W} > 0
$$ 
(several values for $C_{\rm W}$ have been proposed in the literature, see e.g. \cite{DreizlerGross}). From a physical viewpoint, $\rho=v^2$ represents the electronic density (or the electronic density of the valence electrons if the core electrons are already incorporated into $\rho^{\rm nuc}$). The first two terms of $E_\cR^{\rm TFW}(\rho^{\rm nuc}, v)$ model the kinetic energy per simulation cell and the third term the Coulomb energy of the total $\cR$-periodic charge distribution $\rho^{\rm tot}=\rho^{\rm nuc}-v^2$.

The electronic ground state with $Q$ electrons in the simulation cell is obtained by solving the minimization problem
\begin{equation}
\label{eq:TFWgen}
 I_\cR(\rho^{\rm nuc},Q) = \inf \left\{ E_\cR^{\rm TFW}(\rho^{\rm nuc}, v), \; v \in H^1_{\rm per}(\Gamma),\; \int_{\Gamma} v^2 = Q \right\}.
\end{equation}

For the sake of simplicity, we assume that the nuclear charge density is in $L^2_{\rm per}(\Gamma)$. This allows us to gather all the Coulomb interactions in a single, non-negative term (the third term in the right hand side of (\ref{eq:TFWenergyper})). On the other hand, this excludes point-like charges represented by Dirac measures. As often in this field, it is however easy to extend our analysis to point-like nuclei, by splitting the Dirac measure $\delta_0$ as $\delta_0=(\delta_0-\phi)+\phi$ where $\phi$ is a radial function of $C^\infty_c(\R^3)$ such that $\int_{\R^3}\phi=1$ and $\mbox{Supp}(\phi)$ small enough.

The following result is classical. We will however provide a proof of it in Section~\ref{sec:proofs} for the sake of completeness.

\medskip

\begin{prop} \label{prop:EU} Let $\rho^{\rm nuc} \in L^2_{\rm per}(\Gamma)$ and $Q \ge 0$. 
  \begin{enumerate}
  \item Problem (\ref{eq:TFWgen}) has a minimizer $u$ such that $u\in H^4_{\rm per}(\Gamma) \hookrightarrow C^{2}(\R^3) \cap L^{\infty}(\R^3)$ and $u > 0$ in $\R^3$. The function $u$ satisfies the Euler equation 
\begin{equation} \label{eq:Euler}
-C_W \Delta u +\frac{5}{3} C_{\rm TF} u^{7/3} + \left( G_\cR \star_\cR (u^2 -\rho^{\rm nuc})\right) u = \epsilon_{\rm F} u,
\end{equation}
where $\epsilon_{\rm F}$ is the Lagrange multiplier of the constraint $\int_\Gamma u^2 = Q$.
  \item Problem (\ref{eq:TFWgen}) has exactly two minimizers: $u$ and $-u$.
  \end{enumerate}
\end{prop}

\medskip

As a consequence of Proposition~\ref{prop:EU}, the ground state electronic density is always uniquely defined in the framework of the periodic TFW model.  

\section{The Thomas-Fermi-von Weiszäcker model for crystals}
\label{sec:defect}

We now focus on the special case of crystals. More precisely, we consider two kind of systems:
\begin{itemize}
\item a reference $\cR_1$-periodic perfect crystal with nuclear distribution 
$$
\rho^{\rm nuc}_{\rm per} \in L^2_{\rm per}(\Gamma_1),
$$
where $\Gamma_1$ is a unit cell for $\cR_1$;
\item a perturbation of the previous system characterized by the nuclear distribution 
\begin{equation} \label{eq:nuclearCD}
\rho^{\rm nuc} = \rho^{\rm nuc}_{\rm per} + \nu \quad \mbox{with} \quad 
\nu \in {\mathcal C},
\end{equation}
${\mathcal C}$ denoting the Coulomb space defined by (\ref{eq:defC}).
\end{itemize}

\subsection{Reference perfect crystal}
\label{sec:perfect}

It is shown in \cite{Cat98} that the ground state electronic density $\rho^0_{\rm per}$ of a crystal with nuclear charge distribution $\rho^{\rm nuc}_{\rm per} \in L^2_{\rm per}(\Gamma_1)$ can be identified by a thermodynamic limit argument. It is given by $\rho^0_{\rm per} = |u^0_{\rm per}|^2$ where $u^0_{\rm per} \ge 0$ is obtained by solving the minimization problem
\begin{equation}
\label{eq:TFWgen1}
 I_{\cR_1}(\rho^{\rm nuc}_{\rm per},Z) = \inf \left\{E_{\cR_1}^{\rm TFW}(\rho^{\rm nuc}_{\rm per}, v), \; v \in H^1_{\rm per}(\Gamma_1),\; \int_{\Gamma_1} v^2 = Z \right\},
\end{equation}
where
\begin{equation} \label{eq:neutral}
Z = \int_{\Gamma_1} \rho^{\rm nuc}_{\rm per}.
\end{equation}
Note that problem (\ref{eq:TFWgen1}) has a unique solution (up to the sign) for any value of $Z$. The correct value of $Z$ given by (\ref{eq:neutral}) is obtained in \cite{Cat98} by a thermodynamic limit argument. As expected, this value implies the charge neutrality condition 
\begin{equation} \label{eq:neutral2}
\int_{\Gamma_1} (\rho^{\rm nuc}_{\rm per}-\rho^0_{\rm per}) = 0.
\end{equation}
The unique non-negative solution $u^0_{\rm per}$ to (\ref{eq:TFWgen1})-(\ref{eq:neutral}) satisfies the Euler equation
\begin{equation} \label{eq:Eulerperfect}
-C_{\rm W}\Delta u^0_{\rm per} + 
 \frac{5}{3}C_{\rm TF}(\rho^0_{\rm per})^{2/3} u^0_{\rm per} + \left( G_{\cR_1} \star_{\cR_1} (\rho^0_{\rm per}-\rho^{\rm nuc}_{\rm per}) \right) u^0_{\rm per} =  \epsilon^0_{\rm F} u^0_{\rm per},
\end{equation}
where $\epsilon^0_{\rm F}$, the Lagrange multiplier of the charge constraint, which is uniquely defined, is called the Fermi level of the crystal. From (\ref{eq:neutral2}), we infer that  the Coulomb potential $V^0_{\rm per}=G_{\cR_1} \star_{\cR_1} (\rho^0_{\rm per}-\rho^{\rm nuc}_{\rm per})$ is the unique solution in $H^1_{\rm per}(\Gamma_1)$ to the $\cR_1$-periodic Poisson problem
$$
\left\{
\begin{array}{c}
\dps -\Delta V^0_{\rm per}  = \frac{4\pi}{|\Gamma_1|} \left( \rho^0_{\rm per} - \rho^{\rm nuc}_{\rm per} \right), \\
\dps V^0_{\rm per} \; \cR_1\mbox{-periodic}, \; \int_{\Gamma_1} V^0_{\rm per} = 0.\\
\end{array}\right.
$$
By elliptic regularity, $V^0_{\rm per} \in H^2_{\rm per}(\Gamma_1) \hookrightarrow C^0(\R^3) \cap L^\infty(\R^3)$. Using Proposition~\ref{prop:EU}, we obtain that $u^0_{\rm per} \in C^{2}(\R^3) \cap L^{\infty}(\R^3)$, and that $u^0_{\rm per}>0$ in $\mathbb{R}^3$. We thus have the following bounds, that will be useful in our analysis:
\begin{equation} \label{eq:boundu0per}
\exists 0 < m \le M < +\infty \quad \mbox{s.t.} \quad \forall x \in \R^3, \quad m \le u^0_{\rm per}(x) \le M.
\end{equation}
Let us denote by $H^0_{\rm per}$ the periodic Schrödinger operator on $L^2(\R^3)$ with domain $H^2(\R^3)$ and form domain $H^1(\R^3)$ defined by
$$
\forall v \in H^2(\R^3), \quad H^0_{\rm per}v = 
-C_{\rm W}\Delta v +  \frac{5}{3}C_{\rm TF}(\rho^0_{\rm per})^{2/3} v + V^0_{\rm per} v.
$$
It is classical (see e.g. \cite{ReeSim4}) that $H^0_{\rm per}$ is self-adjoint and bounded from below, and that its spectrum is purely absolutely continuous and made of a union of bands. For convenience, we will use the abuse of notation consisting in denoting by $H^0_{\rm per}v$ the distribution
$$
H^0_{\rm per}v := 
-C_{\rm W}\Delta v +  \frac{5}{3}C_{\rm TF}(\rho^0_{\rm per})^{2/3}v + V^0_{\rm per} v,
$$
which is well-defined for any $v \in L^1_{\rm loc}(\R^3)$, and belongs to $H^{-1}(\R^3)$ if $v \in H^1(\R^3)$ and to $H^{-1}_{\rm per}(\Gamma)$ if $v \in H^1_{\rm per}(\Gamma)$. We can thus rewrite equation (\ref{eq:Eulerperfect}) under the form
\begin{equation}
 \label{eq:uper}
H^0_{\rm per} u^0_{\rm per} = \epsilon^0_F u^0_{\rm per}.
\end{equation}
Using the fact that $u^0_{\rm per } > 0$ in $\R^3$, it is easy to see that $\epsilon^0_{\rm F}$ is in fact the minimum of the spectrum of the periodic Schrödinger operator $H^0_{\rm per}$ (that is the bottom of the lowest energy band). As a consequence,
\begin{equation} \label{eq:boundH0per-epsilon0F}
\forall v \in H^1(\R^3), \quad \langle (H^0_{\rm per}-\epsilon^0_{\rm F})v,v \rangle_{H^{-1}(\R^3),H^1(\R^3)} \ge 0.
\end{equation}

\subsection{Crystals with local defects} 
\label{sec:TFWdef}

We now consider a crystal with a local defect whose nuclear charge distribution is given by (\ref{eq:nuclearCD}). It is convenient to describe the TFW electronic state of this system by a function $v$ related to the electronic density $\rho$ by the relation
$$
v = \sqrt{\rho}-u^0_{\rm per}.
$$
We denote by $\mathcal{C}$ the Coulomb space defined as 
\begin{equation} \label{eq:defC}
\mathcal{C} := \left\{ f \in \mathcal{S}'(\mathbb{R}^3) \; | \; \widehat f \in L^1_{\rm loc}(\R^3), \; |\cdot|^{-1} \widehat f(\cdot) \in L^2(\mathbb{R}^3) \right\},
\end{equation}
where $\widehat{f}$ is the Fourier transform of $f$, normalized in such a way that $\|\widehat{f}\|_{L^2(\R^3)} = \|f\|_{L^2(\R^3)}$ for all $f \in L^2(\R^3)$. Endowed with the inner product
$$
D(f,g) := 4 \pi \int_{\mathbb{R}^3} \frac{\overline{\widehat{f}(k)} \, \widehat{g}(k)}{|k|^2} \, dk,
$$
${\mathcal C}$ is a Hilbert space. It holds $L^{6/5}(\R^3) \subset {\mathcal C}$ and
$$
\forall (f,g) \in L^{6/5}(\R^3) \times L^{6/5}(\R^3), \quad
D(f,g) = \int_{\R^3}\int_{\R^3} \frac{f(x) \, g(x')}{|x-x'|} \, dx \, dx'.
$$
Denoting by 
$$
E^{\rm TFW}(\rho^{\rm nuc},w)=C_{\rm W} \int_{\R^3} |\nabla w|^2+ C_{\rm TF} \int_{\R^3} |w|^{10/3}
+ \frac 12 D(\rho^{\rm nuc}-w^2,\rho^{\rm nuc}-w^2)
$$
the TFW energy functional of a finite molecular system {\it in vacuo} with nuclear charge $\rho^{\rm nuc}$, we can formally define the relative energy (with respect to the perfect crystal) of the system with nuclear charge density $\rho^{\rm nuc}_{\rm per}+\nu$ and electronic density $\rho=(u^0_{\rm per}+v)^2$ as
\begin{eqnarray}
&& \!\!\!\!\!\!\!\!\!\!\!\!\!\!\!\!\!\!\!
 E^{\rm TFW}(\rho^{\rm nuc}_{\rm per}+\nu,u^0_{\rm per}+v)-E^{\rm TFW}(\rho^{\rm nuc}_{\rm per},u^0_{\rm per}) \nonumber \\
&=& \langle (H^0_{\rm per}-\epsilon^0_{\rm F})v,v \rangle +  C_{\rm TF} \int_{\R^3} \left( |u^0_{\rm per}+v|^{10/3} - |u^0_{\rm per}|^{10/3} - \frac 53|u^0_{\rm per}|^{4/3}(2u^0_{\rm per}v+v^2) \right) \nonumber \\
&& + \frac 12 D\left(2u^0_{\rm per}v+v^2-\nu,2u^0_{\rm per}v+v^2-\nu\right)
- \int_{\R^3} \nu V^0_{\rm per} + \epsilon^0_{\rm F} q, \label{eq:formal}
\end{eqnarray}
where
\begin{equation} \label{eq:defq}
q= \int_{\R^3} \left( |u^0_{\rm per}+v|^2 - |u^0_{\rm per}|^2\right).
\end{equation}
Of course, the left-hand side of (\ref{eq:formal}) is a formal expression since it is the difference of two quantities taking the value plus infinity. On the other hand, the right-hand side of (\ref{eq:formal}) is mathematically well-defined as soon as $q$ is a fixed real number and $v \in {\mathcal Q}_+$, where
$$
\cQ_+ := \left\{ v\in H^1(\mathbb{R}^3) \; | \; v \geq -u^0_{\rm per} , \; u^0_{\rm per}v \in \mathcal{C} \right\}.
$$
The set ${\mathcal Q}_+$ is a closed convex subset of the Hilbert space
$$
{\mathcal Q} := \left\{v \in H^1(\R^3) \; | \; u^0_{\rm per}v \in {\mathcal C} \right\},
$$
endowed with the inner product defined by
$$
(v,w)_{\mathcal Q} := (v,w)_{H^1(\R^3)} + D(u^0_{\rm per}v,u^0_{\rm per}w).
$$

This formal analysis leads us to propose the following model, which will be justified in the following section by means of thermodynamic limit arguments: the ground state electronic density of the perturbed crystal characterized by the nuclear charge density (\ref{eq:nuclearCD}) is given by
$$
\rho_\nu = (u^0_{\rm per}+v_\nu)^2,
$$
where $v_\nu$ is a minimizer of
\begin{equation} \label{eq:minpbInu}
I^\nu = \inf \left\{ {\mathcal E}^\nu(v), \; v \in {\mathcal Q}_+ \right\}
\end{equation}
with
\begin{eqnarray}
{\mathcal E}^\nu (v) &:=& \langle (H^0_{\rm per}-\epsilon^0_{\rm F})v,v \rangle_{H^{-1}(\R^3),H^1(\R^3)} \nonumber \\ && 
+ C_{\rm TF} \int_{\R^3} \left( |u^0_{\rm per}+v|^{10/3} - |u^0_{\rm per}|^{10/3} - \frac 53 |u^0_{\rm per}|^{4/3}(2u^0_{\rm per}v+v^2) \right) \nonumber \\
&& + \frac 12 D\left(2u^0_{\rm per}v+v^2-\nu,2u^0_{\rm per}v+v^2-\nu\right).
\label{eq:defEnu}
\end{eqnarray}

\medskip

The following result, whose proof is postponed until Section~\ref{sec:proofs}, shows that our model is well-posed.

\medskip

\begin{thm}
\label{Th:defect}
Let $\nu\in {\mathcal C}$. Then,
\begin{enumerate}
 \item Problem (\ref{eq:minpbInu}) has a unique minimizer $v_\nu$, and there exists a positive constant $C_0 > 0$ such that
\begin{equation} \label{eq:boundvnu}
\forall \nu \in {\mathcal C}, \quad \|v_\nu\|_{\mathcal Q} \le C_0 \left( \|\nu\|_{\mathcal C} + \|\nu\|_{\mathcal C}^2 \right).
\end{equation}
The function $v_\nu$ satisfies the Euler equation
\begin{eqnarray}
&& (H^0_{\rm per} - \epsilon^0_F) v_\nu + \frac{5}{3}C_{\rm TF} \left( |u^0_{\rm per}+v_\nu|^{7/3} - |u^0_{\rm per}|^{7/3} - |u^0_{\rm per}|^{4/3}v_\nu \right) \nonumber \\
&& \qquad\qquad\qquad + \left((2u^0_{\rm per}v_\nu + v_\nu^2 - \nu)\star |\cdot|^{-1}\right)(u^0_{\rm per}+v_\nu) = 0. 
\label{eq:eulerpb}
\end{eqnarray}

\item Let us denote by $\rho^0_\nu = \nu-(2u^0_{\rm per}v_\nu + v_\nu^2)$ the total density of charge of the defect and by 
$\Phi^0_\nu = \rho^0_\nu \star |\cdot|^{-1}$ the Coulomb potential generated by $\rho^0_\nu$. It holds $v_\nu \in H^2(\R^3)$, $\Phi^0_\nu \in L^2(\R^3)$ and
\begin{equation} \label{eq:LebesgueValue}
\lim_{r \rightarrow 0} \frac{1}{|B_r|} \int_{B_r} |\widehat{\rho^0_\nu}(k)| \, dk = 0.
\end{equation}
\item
Any minimizing sequence $(v_n)_{n \in \N}$ for (\ref{eq:minpbInu}) converges to $v_\nu$ weakly in $H^1(\R^3)$ and strongly in $L^p_{\rm loc}(\R^3)$ for $1~\le~p~<~6$. Besides, $(u^0_{\rm per}v_n)_{n \in \N}$ converges to $u^0_{\rm per}v_\nu$ weakly in~$\cC$.

For any $q \in \R$, there exists a minimizing sequence $(v_n)_{n \in \N}$ for (\ref{eq:minpbInu}) consisting of functions of $\cQ_+ \cap L^1(\R^3)$ such that
\begin{equation} \label{eq:CCmin}
\forall n \in \N, \quad \int_{\R^3} \left( |u^0_{\rm per}+v_n|^2 - |u^0_{\rm per}|^2\right)=q.
\end{equation}
\end{enumerate}
\end{thm} 

\medskip

We conclude this present section with some physical considerations regarding the charge of the defect.

\medskip

\begin{remark}
Let $\nu \in L^1(\R^3) \cap L^2(\R^3)$. Assuming that $v_\nu \in L^1(\R^3) \cap L^2(\R^3)$ (a property satisfied at least in the special case of a homogeneous host crystal, see Section~\ref{sec:HEG}), then $\widehat{\rho^0_\nu} \in C^0(\R^3)$ and (\ref{eq:LebesgueValue}) simply means that the continuous function $\widehat \rho^0_\nu$ vanishes at $k=0$, or equivalently that
\begin{equation} \label{eq:neutralityrho}
\int_{\R^3} \rho^0_\nu = 0.
\end{equation}
The property (\ref{eq:LebesgueValue}) means that $0$ is a Lebesgue point of $\widehat{\rho^0_\nu}$ and that the Lebesgue value of $\widehat{\rho^0_\nu}$ at $0$ is equal to zero. It can therefore be interpreted as a weak form of the neutrality condition (\ref{eq:neutralityrho}), also valid when $\rho^0_\nu \notin L^1(\R^3)$. The fourth statement of Theorem~\ref{Th:defect} implies that there is no way to model a charge defect within the TFW theory: loosely speaking, if we try to put too many (or not enough) electrons in the system, the electronic density will relax to $(u^0_{\rm per}+v_\nu)^2$ and the remaining (or missing) $q-\int_{\R^3} \nu$ electrons will escape to (or come from) infinity with an energy $\epsilon^0_{\rm F}$. 
\end{remark}

\subsection{Thermodynamic limit}
\label{sec:TL}

The purpose of this section is to provide a mathematical justification of the model (\ref{eq:minpbInu}). Consider a crystal with a local defect characterized by the nuclear charge distribution 
\begin{equation}
\rho^{\rm nuc}=\rho^{\rm nuc}_{\rm per}+\nu \quad \mbox{with} \quad \nu \in L^1(\R^3) \cap L^2(\R^3).
\end{equation}
In numerical simulations, the TFW ground state electronic density of such a system is usually computed with the supercell method. For a given $L \in \N$ large enough, the supercell model of size $L$ is the periodic TFW model (\ref{eq:TFWgen}) with 
\begin{equation} \label{eq:supercellL}
\cR=\cR_L:=L\cR_1, \quad \Gamma = \Gamma_L:=L\Gamma_1, \quad \rho^{\rm nuc}=\rho^{\rm nuc}_{\rm per}+\nu_L, \quad
Q= Z \, L^3 + q,
\end{equation}
where
$$
\nu_L(x) = \sum_{z\in \mathcal{R}_L} (\chi_{\Gamma_L} \nu)(x-z),
$$
$\chi_{\Gamma_L}: \mathbb{R}^3 \rightarrow \mathbb{R}$ denoting the characteristic function of the simulation cell $\Gamma_L$. Note that $\nu_L$ is the unique $\cR_L$-periodic function such that $\nu_L|_{\Gamma_L}=\nu|_{\Gamma_L}$. In practice, $L$ is chosen as large as possible (given the computational means available) to limit the error originating from the artificial Born-von Karman periodic boundary conditions.

\medskip

It is important to note that $u^0_{\rm per}$ is the unique minimizer (up to the sign) of the supercell model of size $L$ for $\rho^{\rm nuc}=\rho^{\rm nuc}_{\rm per}$ and $Q=ZL^3$, whatever $L \in \N^\ast$. Reasoning as in the previous section, we introduce the energy functional
\begin{eqnarray}
{\mathcal E}^{\nu}_L (v_L) &:=& \langle (H^0_{\rm per}-\epsilon^0_{\rm F})v_L,v_L \rangle_{H^{-1}_{\rm per}(\Gamma_L),H^1_{\rm per}(\Gamma_L)}  \nonumber \\ &&
+ C_{\rm TF} \int_{\Gamma_L} \left( |u^0_{\rm per}+v_L|^{10/3} - |u^0_{\rm per}|^{10/3} - \frac 53 |u^0_{\rm per}|^{4/3}(2u^0_{\rm per}v_L+v_L^2) \right) \nonumber \\
&& + \frac 12 D_{\cR_L}\left(2u^0_{\rm per}v_L+v_L^2-\nu_L,2u^0_{\rm per}v_L+v_L^2-\nu_L\right),
\label{eq:defEnuL}
\end{eqnarray}
which is such that
\begin{equation}
E^{\rm TFW}_{\cR_L}(\rho^{\rm nuc}_{\rm per}+\nu_L,u^0_{\rm per}+v_L)-E^{\rm TFW}_{\cR_L}(\rho^{\rm nuc}_{\rm per},u^0_{\rm per}) = {\mathcal E}^{\nu}_L (v_L)
- \int_{\Gamma_L} \nu_L V^0_{\rm per} + \epsilon^0_{\rm F} q, \label{eq:diffTFW}
\end{equation}
with
\begin{equation} \label{eq:defqsupercell}
q=\int_{\Gamma_L} \left( |u^0_{\rm per}+v_L|^2 - |u^0_{\rm per}|^2\right)
= \int_{\Gamma_L} (2u^0_{\rm per}v_L+v_L^2).
\end{equation}
While (\ref{eq:formal}) and (\ref{eq:defq}) are formal expressions, (\ref{eq:diffTFW}) and (\ref{eq:defqsupercell}) are well-defined mathematical expressions. The ground state electronic density of the supercell model for the data defined by (\ref{eq:supercellL}) is therefore obtained as
$$
\rho^{0,\nu,q}_L = (u^0_{\rm per}+v_{\nu,q,L})^2
$$
where $v_{\nu,q,L}$ is a minimizer of 
\begin{equation} \label{eq:minTFWdefL}
I^{\nu,q}_L = \inf \left\{ {\mathcal E}^\nu_{L}(v_L), \; v_L \in \cQ_{+,L}, \; \int_{\Gamma_L} (2u^0_{\rm per}v_L + v_L^2) = q \right\},
\end{equation}
$\cQ_{+,L}$ denoting the convex set
$$
\cQ_{+,L}  =  \left\{ v_L\in H^1_{\rm per}(\Gamma_L) \; | \; v_L \geq -u^0_{\rm per} \right\}.
$$
We also introduce the minimization problem
\begin{equation} \label{eq:minTFWdefLf}
I^{\nu}_L = \inf \left\{ {\mathcal E}^\nu_{L}(v_L), \; v_L \in \cQ_{+,L}  \right\},
\end{equation}
in which we do not impose {\it a priori} the electronic charge in the supercell.

\medskip

\begin{thm}
\label{Th:thlim}
Let $\nu \in L^1(\R^3) \cap L^2(\R^3)$.
\begin{enumerate}
 \item {\em Thermodynamic limit with charge constraint.} For each $q\in\mathbb{R}$ and each $L \in \N^\ast$, the minimization problem (\ref{eq:minTFWdefL}) has a unique minimizer $v_{\nu,q,L}$. For each $q \in \R$, the sequence $(v_{\nu,q,L})_{L\in\mathbb{N}^*}$ converges, weakly in $H^1_{\rm loc}(\mathbb{R}^3)$, and strongly in $L^p_{\rm loc}(\mathbb{R}^3)$ for all $1\leq p < 6$, towards $v_\nu$, the unique solution to problem (\ref{eq:minpbInu}). For each  $q\in\mathbb{R}$ and each $L \in \N^\ast$, $v_{\nu,q,L}$ satisfies the Euler equation
 \begin{eqnarray}
&& \!\!\!\!\!\!\!\!\!\!\!\!\!\!\!\!\!\!\!\!\!
(H^0_{\rm per} - \epsilon^0_F) v_{\nu,q,L} + \frac{5}{3}C_{\rm TF} \left( |u^0_{\rm per}+v_{\nu,q,L}|^{7/3} - |u^0_{\rm per}|^{7/3} - |u^0_{\rm per}|^{4/3}v_{\nu,q,L} \right) \nonumber \\
&&  \!\!\!\!\!\!\!\!\!\!\!\!\!\!\!\!\!\!\!\!\! 
+ \left((2u^0_{\rm per}v_{\nu,q,L} + v_{\nu,q,L}^2 - \nu_L)\star_{\cR_L} G_{\cR_L} \right)(u^0_{\rm per}+v_{\nu,q,L}) = \mu_{\nu,q,L}(u^0_{\rm per}+v_{\nu,q,L}), 
\label{eq:eulerpbL}
\end{eqnarray}
where $\mu_{\nu,q,L} \in \R$ is the Lagrange multiplier of the constraint $\int_{\Gamma_L} (2u^0_{\rm per}v_{\nu,q,L}+v_{\nu,q,L}^2) = q$, and it holds $\dps \lim_{L \to \infty} \mu_{\nu,q,L} = 0$ for each $q \in \R$.

\item {\em Thermodynamic limit without charge constraint.} For each $L \in \N^\ast$, the minimization problem (\ref{eq:minTFWdefLf}) has a unique minimizer $v_{\nu,L}$. It holds
 \begin{eqnarray}
&& \!\!\!\!\!\!\!\!\!\!\!\!\!\!\!\!\!\!\!\!\!
(H^0_{\rm per} - \epsilon^0_F) v_{\nu,L} + \frac{5}{3} C_{\rm TF} \left( |u^0_{\rm per}+v_{\nu,L}|^{7/3} - |u^0_{\rm per}|^{7/3} - |u^0_{\rm per}|^{4/3}v_{\nu,L} \right) \nonumber \\
&&  \!\!\!\!\!\!\!\!\!\!\!\!\!\!\!\!\!\!\!\!\! 
+ \left((2u^0_{\rm per}v_{\nu,L} + v_{\nu,L}^2 - \nu_L)\star_{\cR_L} G_{\cR_L} \right)(u^0_{\rm per}+v_{\nu,L}) = 0.
\label{eq:eulerpbL2}
\end{eqnarray}
The sequence $(v_{\nu,L})_{L\in \mathbb{N}^*}$ also converges to $v_\nu$, weakly in $H^1_{\rm loc}(\mathbb{R}^3)$, and strongly in $L^p_{\rm loc}(\mathbb{R}^3)$ for all $1\leq p < 6$. Besides,
$$
\int_{\Gamma_L} \left(\nu_L - (2u^0_{\rm per}v_{\nu,L} + v_{\nu,L}^2)\right) \mathop{\longrightarrow}_{L \to \infty} 0.
$$
\end{enumerate}
\end{thm}

\subsection{The special case of homogeneous host crystals} 
\label{sec:HEG}

In this section, we address the special case when the host crystal is a homogeneous medium completely characterized by the positive real number $\alpha$ such that 
\begin{equation} \label{eq:jellium}
\forall x \in \R^3, \quad 
\rho^{\rm nuc}_{\rm per}(x) =\rho^0_{\rm per}(x) = \alpha^2 \quad \mbox{and} \quad u^0_{\rm per}(x) = \alpha.
\end{equation}
In this case, analytical expressions for the linear response can be derived, leading to the following result.

\medskip

\begin{thm}
 \label{Jellium} Assume that (\ref{eq:jellium}) holds. For each $\nu\in\mathcal{C}$, the unique solution $v_\nu$ to (\ref{eq:minpbInu}) can be expanded as 
\begin{equation} \label{eq:vnu}
v_\nu = g \star \nu + \widetilde r_2(\nu)
\end{equation}
where $g\in L^1(\mathbb{R}^3)$ is characterized by its Fourier transform
$$
\widehat g(k) = \frac{1}{(2\pi)^{3/2}} \, \frac{4\pi\alpha}{C_{\rm W}|k|^4+\frac{20}9 C_{\rm TF}\alpha^{4/3} |k|^2+8\pi\alpha^2},
$$
and where $\widetilde r_2(\nu) \in L^1(\R^3)$. For each $\nu\in L^1(\mathbb{R}^3)\cap\mathcal{C}$, it holds $v_\nu \in L^1(\mathbb{R}^3) \cap L^2(\R^3)$ and
$$
\int_{\mathbb{R}^3} (\nu-(2 u^0_{\rm per}v_\nu + v_\nu^2)) = 0.
$$
\end{thm}

\medskip

The first term in the right hand side of (\ref{eq:vnu}) is in fact the linear component of the application $\nu \mapsto v_\nu$. The second term gathers the higher order contributions.

\medskip

\begin{proof} In the special case under consideration, the Euler equation (\ref{eq:eulerpb}) also reads
\begin{equation} \label{eq:EulerJ}
-C_{\rm W}\Delta v_\nu + \frac{20}{9} C_{\rm TF}\alpha^{4/3}v_\nu + 2 \alpha^2 \left(v_\nu \star |\cdot|^{-1} \right) =  \alpha \left(\nu \star |\cdot|^{-1} \right)  -   \alpha \left(v_\nu^2 \star |\cdot|^{-1} \right) + \kappa_\nu, 
\end{equation}
where
$$
\kappa_\nu = - \frac{5}{3}C_{\rm TF} \left(|\alpha+v_\nu|^{7/3} - \alpha^{7/3} - \frac{7}{3} \alpha^{4/3} v_\nu \right) + \left((\nu - 2\alpha v_\nu-v_\nu^2) \star |\cdot|^{-1} \right)v_\nu.
$$
We therefore obtain (\ref{eq:vnu}) with
$$
\widetilde r_2(\nu) = - g \star v_\nu^2 + h \star \kappa_\nu,
$$
the convolution kernel $h$ being defined through its Fourier transform as
$$
\widehat h(k)  = \frac{1}{(2\pi)^{3/2}} \, \frac{|k|^2}{C_{\rm W}|k|^4+\frac{20}9 C_{\rm TF}\alpha^{4/3} |k|^2+8\pi\alpha^2}.
$$
It follows from the second statement of Theorem~\ref{Th:defect} and Lemma~\ref{lem:convex} below that $\kappa_\nu \in L^1(\R^3)$. The proof will therefore be complete as soon as we have proven that $g \in L^1(\R^3)$ and $h \in L^1(\R^3)$. In fact, we will prove that any even-tempered distribution $f \in {\mathcal S}'(\R^3)$ whose Fourier transform is a function of the form
$$
\widehat{f}(k) = \frac{q(|k|)}{|k|r(|k|)}
$$
where $q$ and $r$ are polynomials of the real variable satisfying $\mbox{deg}(q)<\mbox{deg}(r)$, $r > 0$ in $\R_+$, $q(0)=0$, $q''(0)=0$ and $r'(0)=0$ (which is the case for both $g$ and $h$), is in $L^1(\R^3)$. Indeed, $\widehat f \in L^2(\R^3)$ and a simple calculation shows that 
\begin{eqnarray*}
f(x) & = & \sqrt{\frac{2}\pi} \, \frac{1}{|x|} \int_0^{+\infty} \left(\frac{q}{r}\right)(t) \, \sin(|x|t) \, dt \\
& = & \sqrt{\frac{2}\pi} \, \frac{1}{|x|^2} \int_0^{+\infty} \frac{d}{dt} \left( \frac{q}{r} \right)(t) \,  \cos(|x|t) \, dt \\
& = & \sqrt{\frac{2}\pi} \, \frac{1}{|x|^5} \int_0^{+\infty} \frac{d^4}{dt^4} \left( \frac{q}{r} \right)(t) \,  \sin(|x|t) \, dt.
\end{eqnarray*}
Therefore, there exists $C\in \R_+$ such that
$$
|f(x)| \leq \frac{C}{|x|^2+|x|^5}
$$
almost everywhere in $\R^3$, which proves that $f \in L^1(\R^3)$. 
\end{proof}

\medskip

\begin{remark} For a generic $\nu\in\mathcal{C}$, the function $v_\nu$, hence the density $2\alpha v_\nu +v_\nu^2$, are {\em  not} in $L^1(\mathbb{R}^3)$. This follows from the fact that the nonlinear contribution $\widetilde r_2(\nu)$ is always in $L^1(\R^3)$, while the linear contribution $g \star \nu$ is not necessarily in $L^1(\R^3)$ since its Fourier transform
$$
\widehat{(g \star \nu)}(k) =  \frac{4\pi \alpha}{C_{\rm W}|k|^4 + \frac{20}{9}C_{\rm TF}\alpha^{4/3} |k|^2 + 8\pi\alpha^2} \widehat{\nu}(k)
$$
is not necessarily in $L^{\infty}(\mathbb{R}^3)$.
\end{remark}

\section{Proofs}
\label{sec:proofs}

This section is devoted to the proofs of Proposition~\ref{prop:EU}, Theorem~\ref{Th:defect} and Theorem~\ref{Th:thlim}.

\medskip

In the sequel, we set
$$
C_{\rm TF} =1 \quad \mbox{and} \quad C_{\rm W} = 1 \qquad 
\mbox{(in order to simplify the notation)}.
$$

\subsection{Preliminary results}

We first state and prove a few useful lemmas. Some of these results are simple, or well-known, but we nevertheless prove them here for the sake of self-containment.

\medskip

\begin{lem} \label{lem:convex}
For all $0 < m \le M < \infty$ and all $\gamma \ge 2$, there exists $C \in \R_+$ such that for all $m \le a \le M$ and all $b \ge -a$,
\begin{equation}
(\gamma-1)a^{\gamma-2} b^2 \le (a+b)^{\gamma}-a^{\gamma}-\gamma a^{\gamma-1}b  \le  C \left(1+|b|^{\gamma-2}\right) b^2.
\label{eq:ineq103}
\end{equation}
\end{lem}

\medskip

\begin{proof} Let $\phi(t) = (a+tb)^\gamma$. It holds for all $t \in (0,1)$, $\phi'(t)=\gamma (a+tb)^{\gamma-1} b$ and $\phi''(t) = \gamma(\gamma-1)(a+tb)^{\gamma-2}b^2$. Using the identity
$$
\phi(1)-\phi(0)-\phi'(0) = \int_0^1 (1-t) \phi''(t) \, dt,
$$
we get
\begin{eqnarray*}
(a+b)^{\gamma}-a^{\gamma}-\gamma a^{\gamma-1}b &=& \gamma(\gamma-1) b^2
\int_0^1 (1-t) (a+tb)^{\gamma-2} \, dt .
\end{eqnarray*}
We obtain (\ref{eq:ineq103}) using the fact that for all $t \in [0,1]$, $a(1-t) \le a+tb \le M+|b|$.
\end{proof}

\medskip

\begin{lem}\label{lem:density}
Let $\nu \in {\mathcal C}$ and $v \in {\mathcal Q}_+ \cap H^2(\R^3)$ such that $v > -u^0_{\rm per}$ in $\R^3$. For all $\epsilon > 0$ and $q \in \R$, there exists $v_\epsilon \in {\mathcal Q}_+ \cap C^2_c(\R^3)$ such that
$$
\int_{\R^3} (2u^0_{\rm per}v_\epsilon+v_\epsilon^2)=q \quad \mbox{and} \quad 
|{\mathcal E}^\nu(v_\epsilon)-{\mathcal E}^\nu(v)|\le \epsilon.
$$ 
\end{lem}

\medskip

\begin{proof} Let $\epsilon > 0$.
As the functions of $H^2(\R^3)$ are continuous and decay to zero at infinity, there exists $\delta > 0$ such that 
\begin{equation}\label{eq:vueta}
\forall x \in \R^3, \quad v(x) \ge -u^0_{\rm per}(x)+\delta.
\end{equation}
For all $R>0$, let $B_R$ be the ball of $\R^3$ centered at zero and of radius $R$. For $\eta > 0$, we define
$$
v^\eta = (u^0_{\rm per})^{-1} {\mathcal F}^{-1}\left(\chi_{\overline{B}_{1/\eta}\setminus B_\eta} {\mathcal F}(u^0_{\rm per}v) \right),
$$
where ${\mathcal F}$ is the Fourier transform and ${\mathcal F}^{-1}$ the inverse Fourier transform. Clearly, $v^\eta \in H^4(\R^3) \hookrightarrow C^2(\R^3)$ and $u^0_{\rm per}v^\eta \in {\mathcal C}$. In addition, when $\eta$ goes to zero, $(v^\eta)_{\eta >0}$ converges to $v$ in $H^2(\R^3)$, hence in $L^\infty(\R^3)$, and $(u^0_{\rm per}v^\eta)_{\eta > 0}$ converges to $u^0_{\rm per}v$ in ${\mathcal C}$. The function ${\mathcal E}^\nu$ being continuous on ${\mathcal Q}$, this implies that there exists some $\eta_0 > 0$ such that 
$$
v^{\eta_0} \in {\mathcal Q}_+ \cap C^2(\R^3) \quad \mbox{and} \quad
|{\mathcal E}^\nu(v^{\eta_0})-{\mathcal E}^\nu(v)| \le \epsilon/4.
$$
Let $\chi$ be a function of $C^\infty_c(\R^3)$ supported in $B_2$, such that $0 \le \chi(\cdot) \le 1$ and $\chi=1$ in $B_1$. 
For $n \in \N^\ast$, we denote by $\chi_n(\cdot)=\chi(n^{-1}\cdot)$ and by $v^{\eta_0,n}=\chi_n v^{\eta_0}$. For each $n \in \N^\ast$, $v^{\eta_0,n} \in {\mathcal Q}_+ \cap C^2_c(\R^3)$ and the sequence $(v^{\eta_0,n})_{n \in \N^\ast}$ converges to $v^{\eta_0}$ in ${\mathcal Q}$ when $n$ goes to infinity. Hence, we can find some $n_0 > 0$ such that 
$$
v^{\eta_0,n_0} \in {\mathcal Q}_+ \cap C^2_c(\R^3) \quad \mbox{and} \quad
|{\mathcal E}^\nu(v^{\eta_0,n_0})-{\mathcal E}^\nu(v^{\eta_0})| \le \epsilon/4.
$$
Let 
$$
q_0 = \int_{\R^3} (2u^0_{\rm per}v^{\eta_0,n_0}+(v^{\eta_0,n_0})^2) \quad \mbox{and} \quad q_1=q-q_0.
$$
If $q_1=0$, $v^\epsilon=v^{\eta_0,n_0}$ fulfills the conditions of Lemma~\ref{lem:density}. Otherwise, we introduce for $m$ large enough the function $v_m$ defined as $v_m=t_m \chi_m u^0_{\rm per}$ where $t_m$ is the larger of the two real numbers such that 
$$
\int_{\R^3} (2u^0_{\rm per}v_m+v_m^2) = 2 t_m \int_{\R^3} \chi_m \rho^0_{\rm per} + t_m^2 \int_{\R^3} \chi_m^2 \rho^0_{\rm per} =  q_1.
$$
A simple calculation shows that $t_m \dps \mathop{\sim}_{m \to \infty} \frac{1}{2}q_1  |\Gamma_1|Z^{-1} \left( \int_{\R^3} \chi \right)^{-1} m^{-3}$, and that
$$
\lim_{m \to \infty} {\mathcal E}^0(v_m) = 0,
$$
so that there exists $m_0 \in \N^\ast$ such that $v_m \in {\mathcal Q}_+ \cap C^2_c(\R^3)$ and $0 \le {\mathcal E}^0(v_{m_0}) \le \epsilon/4$. Let us finally choose some $R_1 \in {\mathcal R}_1 \setminus \left\{0\right\}$ and introduce the sequence of functions $(v^{\eta_0,n_0}_{m_0,p})_{p \in \N}$ defined by 
$$
v^{\eta_0,n_0}_{m_0,p}(\cdot)=v^{\eta_0,n_0}(\cdot)+v_{m_0}(\cdot-pR_1).
$$
For $p$ large enough, $v^{\eta_0,n_0}_{m_0,p}$ belongs to ${\mathcal Q}_+ \cap C^2_{\rm c}(\R^3)$ and satisfies 
$$
\int_{\R^3} (2u^0_{\rm per}v^{\eta_0,n_0}_{m_0,p}+(v^{\eta_0,n_0}_{m_0,p})^2) = q.
$$
Besides,
\begin{eqnarray*}
&& |{\mathcal E}^\nu(v^{\eta_0,n_0}_{m_0,p})-{\mathcal E}^\nu(v^{\eta_0,n_0})| \\
&& \quad =
\left| {\mathcal E}^0(v_{m_0}) +D(2u^0_{\rm per}v^{\eta_0,n_0}+(v^{\eta_0,n_0})^2-\nu,(2u^0_{\rm per}v_{m_0}+v_{m_0}^2)(\cdot-pR_1)) \right| \\ && \quad \le 
\epsilon/4 + \left|D(2u^0_{\rm per}v^{\eta_0,n_0}+(v^{\eta_0,n_0})^2-\nu,(2u^0_{\rm per}v_{m_0}+v_{m_0}^2)(\cdot-pR_1)) \right|.
\end{eqnarray*}
As 
$$
\lim_{p \to \infty}D(2u^0_{\rm per}v^{\eta_0,n_0}+(v^{\eta_0,n_0})^2-\nu,(2u^0_{\rm per}v_{m_0}+v_{m_0}^2)(\cdot-pR_1)) = 0,
$$
there exists some $p_0 \in \N$ such that 
$$
\left|D(2u^0_{\rm per}v^{\eta_0,n_0}+(v^{\eta_0,n_0})^2-\nu,(2u^0_{\rm per}v_{m_0}+v_{m_0}^2)(\cdot-pR_1)) \right| \le \epsilon /4.
$$
Setting $v^\epsilon=v^{\eta_0,n_0}_{m_0,p_0}$, we get the desired result.
\end{proof}

\medskip

The next four lemmas are useful to pass to the thermodynamic limit in the Coulomb term (Lemmas~\ref{lem:boundDL},~\ref{lem:limitDL} and~\ref{lem:CVDL}) and in the kinetic energy term (Lemma~\ref{lem:weakiL}). 

\medskip

\begin{lem} \label{lem:boundDL}
There exists a constant $C \in \R_+$ such that for all $L \in \N^\ast$,
\begin{eqnarray*}
\forall  \rho_L \in L^1_{\rm per}(\Gamma_L) \cap L^{6/5}_{\rm per}(\Gamma_L), & \; & D_{\cR_L}(\rho_L,\rho_L) \le C \left( \|\rho_L\|_{L^1_{\rm per}(\Gamma_L)}^2+\|\rho_L\|_{L^{6/5}_{\rm per}(\Gamma_L)}^2 \right), \\
\forall  v_L \in H^1_{\rm per}(\Gamma_L), & \; & D_{\cR_L}(v_L^2,v_L^2) \le C \|v_L\|_{H^1_{\rm per}(\Gamma_L)}^4.
\end{eqnarray*}
\end{lem}

\medskip

\begin{proof} It is well-known (see e.g. \cite{Cat98}) that 
$$
\forall x \in \Gamma_1, \quad G_{\cR_1}(x) = |x|^{-1} + g(x), 
$$
with $g \in L^\infty(\Gamma_1)$, and that for all $L \in \N^\ast$,
$$
\forall x \in \R^3, \quad G_{\cR_L}(x) = L^{-1} G_{\cR_1}(L^{-1}x).
$$
Let ${\mathcal I} = \left\{R \in {\mathcal R}_1 \, | \,  \exists (x,y) \in \overline{\Gamma}_1 \times \overline{\Gamma}_1 \mbox{ s.t. } x-y = R \right\}$. It holds
$$
\forall (x,y) \in \Gamma_L \times \Gamma_L, \quad 
0 \le G_{\cR_L}(x-y) \le \sum_{R \in {\mathcal I}} |x-y-LR|^{-1} + L^{-1} \|g\|_{L^\infty}.  
$$
Therefore, for all $L \in \N^\ast$,
\begin{eqnarray*}
D_{\cR_L}(\rho_L,\rho_L) & = & \int_{\Gamma_L} \int_{\Gamma_L} G_{\cR_L}(x-y) \rho_L(x) \rho_L(y) \, dx \, dy \\
& \le & \sum_{R \in {\mathcal I}} \int_{\R^3} \int_{\R^3} \frac{\chi_{\Gamma_L}(x)|\rho_L(x)| \, \chi_{\Gamma_L}(y) |\rho_L(y)|}{|x-y-LR|} \, dx \, dy + L^{-1} \|g\|_{L^\infty} \|\rho_L\|_{L^1_{\rm per}(\Gamma_L)}^2 \\
& \le & C' \|\chi_{\Gamma_L}\rho_L\|_{L^{6/5}(\R^3)}^2 + \|g\|_{L^\infty} \|\rho_L\|_{L^1_{\rm per}(\Gamma_L)}^2 \\
& = & C' \|\rho_L\|_{L^{6/5}_{\rm per}(\Gamma_L)}^2 + \|g\|_{L^\infty} \|\rho_L\|_{L^1_{\rm per}(\Gamma_L)}^2,
\end{eqnarray*}
where $C'$ is a constant independent of $L$ and $\rho_L$. Let $C_1$ be the Sobolev constant such that
$$
\forall v_1 \in H^1_{\rm per}(\Gamma_1), \quad \|v_1\|_{L^6_{\rm per}(\Gamma_1)} \le C_1 \|v_1\|_{H^1_{\rm per}(\Gamma_1)}.
$$
By an elementary scaling argument, it is easy to check that the inequality 
$$
\forall v_L \in H^1_{\rm per}(\Gamma_L), \quad \|v_L\|_{L^6_{\rm per}(\Gamma_L)} \le C_1 \|v_L\|_{H^1_{\rm per}(\Gamma_L)}
$$
holds for all $L \in \N^\ast$.
Thus, for all $v_L \in H^1_{\rm per}(\Gamma_L)$, we obtain
$$
\|v_L^2\|_{L^{6/5}_{\rm per}(\Gamma_L)}^2
= \|v_L\|_{L^{12/5}_{\rm per}(\Gamma_L)}^4 \le \|v_L\|_{L^{2}_{\rm per}(\Gamma_L)}^3 \|v_L\|_{L^{6}_{\rm per}(\Gamma_L)} \le C_1 \|v_L\|_{H^1_{\rm per}(\Gamma_L)}^4,
$$
which completes the proof of Lemma~\ref{lem:boundDL}.
\end{proof}

\medskip

\begin{lem} \label{lem:limitDL}
Let $\nu \in L^1(\R^3) \cap L^2(\R^3)$ and $\nu_L \in L^2_{\rm per}(\Gamma_L)$ defined by $\nu_L|_{\Gamma_L}=\nu|_{\Gamma_L}$ for all $L\in \mathbb{N}^*$. Then
\begin{equation} \label{eq:limDL}
\lim_{L \to \infty} D_{\cR_L}(\nu_L,\nu_L) = D(\nu,\nu).
\end{equation}
\end{lem}

\medskip

\begin{proof} Let  $g_1 := \dps |\Gamma_1|^{-1}\int_{\Gamma_1} G_1$ and $\Gamma^\ast_L$ be the first Brillouin zone of the lattice ${\mathcal R}_L$ (that is the Voronoi cell of the origin in the dual space). Note that ${\mathcal R}_L^\ast=L^{-1} {\mathcal R}_1^\ast$ and $\Gamma^\ast_L=L^{-1}\Gamma^\ast_1$. Let $K > 0$. We have
\begin{eqnarray}
D_{\cR_L}(\nu_L,\nu_L) &=& g_1 L^{-1} \left( \int_{\Gamma_L} \nu \right)^2
+ \sum_{k \in L^{-1}\cR_1^\ast \setminus \left\{0\right\}} \frac{4\pi}{|k|^2}
|c_{k,L}(\nu_L)|^2 \nonumber \\&=& g_1 L^{-1} \left( \int_{\Gamma_L} \nu \right)^2
+ 4\pi \sum_{k \in B_{K} \cap L^{-1}\cR_1^\ast \setminus \left\{0\right\}} |\Gamma_L^\ast| \frac{|\widetilde c_{k,L}(\nu_L)|^2}{|k|^2} \nonumber \\
&& + 4\pi \sum_{k \in B_{K}^c \cap L^{-1}\cR_1^\ast \setminus \left\{0\right\}} \frac{|c_{k,L}(\nu_L)|^2}{|k|^2} ,  \label{eq:Riemann}
\end{eqnarray}
where $B_K$ is the ball of radius $K$ centered at $0$, $B_K^c = \R^3 \setminus \overline{B}_K$, 
$$
c_{k,L}(\nu_L)= |\Gamma_L|^{-1/2} \int_{\Gamma_L} \nu_L(x)e^{-ik\cdot x} \, dx,
$$ 
and 
$$
\widetilde c_{k,L}(\nu_L) =  |\Gamma_L^\ast|^{-1/2} c_{k,L}(\nu_L) =  \frac{1}{(2\pi)^{3/2}} \int_{\Gamma_L} \nu(x) e^{-ik\cdot x} \, dx.
$$
As $\nu \in L^1(\R^3)$, $|\widetilde c_{k,L}(\nu_L)| \le (2\pi)^{-3/2} \|\nu\|_{L^1(\R^3)}$ for all $k$ and $L$, $\widehat \nu \in L^\infty(\R^3)$, and
$$
\forall k \in \R^3, \quad \widetilde c_{k,L}(\nu_L) \quad \mathop{\longrightarrow}_{L \to \infty} \quad \widehat \nu(k)
$$
Clearly the first term in the right hand side of (\ref{eq:Riemann}) goes to zero when $L$ goes to infinity. Besides,
$$
 \sum_{k \in B_{K} \cap L^{-1}\cR_1^\ast \setminus \left\{0\right\}} |\Gamma_L^\ast| \frac{|\widetilde c_{k,L}(\nu_L)|^2}{|k|^2}  \quad \mathop{\longrightarrow}_{L \to \infty} \quad  \int_{B_K} \frac{|\widehat\nu(k)|^2}{|k|^2}  \, dk.
$$
Lastly,
\begin{eqnarray*}
\sum_{k \in B_{K}^c \cap L^{-1}\cR_1^\ast \setminus \left\{0\right\}} \frac{| c_{k,L}(\nu_L)|^2}{|k|^2}  & \le & 
\left( \sum_{k \in B_{K}^c \cap L^{-1}\cR_1^\ast \setminus \left\{0\right\}} \frac{|c_{k,L}(\nu_L)|^2}{|k|^4}\right)^{1/2} \left( \sum_{k \in B_{K}^c \cap L^{-1}\cR_1^\ast \setminus \left\{0\right\}}  |c_{k,L}(\nu_L)|^2 \right)^{1/2} \\
& \le & \frac{1}{(2\pi)^{3/2}} 
\left( \sum_{k \in B_{K}^c \cap L^{-1}\cR_1^\ast \setminus \left\{0\right\}} |\Gamma_L^\ast| \frac{1}{|k|^4}\right)^{1/2} \|\nu\|_{L^1(\R^3)} \|\nu\|_{L^2(\R^3)} \\
&& \mathop{\longrightarrow}_{L \to \infty} \quad \frac{1}{(2\pi^2K)^{1/2}} \|\nu\|_{L^1(\R^3)} \|\nu\|_{L^2(\R^3)}.
\end{eqnarray*}
It is then easy to conclude that (\ref{eq:limDL}) holds true.
\end{proof}

\medskip

\begin{lem} \label{lem:CVDL} Let $(\rho_L)_{L \in \N^\ast}$ be a sequence of functions of $L^2_{\rm loc}(\R^3)$ such that 
  \begin{enumerate}
  \item for each $L \in \N^\ast$, $\rho_L \in L^2_{\rm per}(\Gamma_L)$;
  \item there exists $C \in \R_+$ such that for all $L \in \N^\ast$, 
    $$
    \left| \int_{\Gamma_L}\rho_L \right| \le C \quad \mbox{and} \quad D_{\cR_L}(\rho_L,\rho_L) \le C;
    $$
  \item there exists $\rho \in {\mathcal D}'(\R^3)$ such that $(\rho_L)_{L \in \N^\ast}$ converges to $\rho$ in ${\mathcal D}'(\R^3)$. 
  \end{enumerate}
Then $\rho \in {\mathcal C}$ and
\begin{equation} \label{eq:liminfD}
D(\rho,\rho) \le \liminf_{L\to \infty}  D_{\cR_L}(\rho_L,\rho_L).
\end{equation}

\medskip

\noindent
In addition, for any $p > 6/5$ and any sequence $(v_L)_{L \in \N^*}$ of functions of $L^p_{\rm loc}(\R^3)$ such that $v_L \in L^p_{\rm per}(\Gamma_L)$ for all $L \in \N^\ast$, which weakly converges to some $v \in L^p_{\rm loc}(\R^3)$ in $L^p_{\rm loc}(\R^3)$, it holds
\begin{equation} \label{eq:CVD}
\forall \phi \in C^\infty_c(\R^3), \quad \lim_{L \to \infty} 
D_{\cR_L}(\rho_L,v_L\phi) = D(\rho,v\phi).
\end{equation}
\end{lem}

\medskip

\begin{proof} Let $W_L$ the unique solution in $H^2_{\rm per}(\Gamma_L)$ to
\begin{equation}\label{eq:defWL}
\left\{ \begin{array}{l}
\dps -\Delta W_L = 4\pi \left( \rho_L - |\Gamma_L|^{-1} \int_{\Gamma_L}\rho_L \right) \\
\dps W_L \mbox{ $\cR_L$-periodic}, \quad \int_{\Gamma_L} W_L = 0. \end{array} \right.
\end{equation}
It holds
\begin{equation}\label{eq:boundWL}
\frac{1}{4\pi} \int_{\Gamma_L} |\nabla W_L|^2 = D_{\cR_L}(\rho_L,\rho_L) - g_1L^{-1} \left( \int_{\Gamma_L}\rho_L \right)^2 \le C,
\end{equation}
where $g_1 := \dps |\Gamma_1|^{-1}\int_{\Gamma_1} G_{\cR_1} \geq 0$. Hence the sequence $(\|\nabla W_L\|_{L^2_{\rm per}(\Gamma_L)})_{L \in \N^*}$ is bounded.

By Sobolev and Poincaré-Wirtinger inequalities, we have
$$
\forall V_1 \in H^1_{\rm per}(\Gamma_1) \mbox{ s.t. } \int_{\Gamma_1} V_1 = 0, \quad  \|V_1\|_{L^6_{\rm per}(\Gamma_1)} \le C_1 \|V_1\|_{H^1_{\rm per}(\Gamma_1)} \le C'_1 \| \nabla V_1 \|_{L^2_{\rm per}(\Gamma_1)},
$$ 
and by a scaling argument, we obtain that for all $L \in \N^\ast$, 
$$
\forall V_L \in H^1_{\rm per}(\Gamma_L) \mbox{ s.t. } \int_{\Gamma_L} V_L = 0, \quad   \|V_L\|_{L^6_{\rm per}(\Gamma_L)}  \le C'_1 \| \nabla V_L \|_{L^2_{\rm per}(\Gamma_L)},
$$
where the constant $C'_1$ does not depend on $L$. Thus, the sequence $(\|W_L\|_{L^6_{\rm per}(\Gamma_L)})_{L \in \N^*}$ is bounded. Let $\widetilde C \in \R_+$ such that 
$$
\forall L \in \N^\ast, \quad \|W_L\|_{L^6_{\rm per}(\Gamma_L)} \le \widetilde C \quad \mbox{and} \quad \|\nabla W_L\|_{L^2_{\rm per}(\Gamma_L)} \le \widetilde C, 
$$
and let $(R_n)_{n \in \N}$ be an increasing sequence of positive real numbers such that $\lim_{n \to \infty} R_n = \infty$. 
Let $R>0$. For $L > 2R$, 
$$
\| W_L \|_{L^6(B_R)} \le \|W_L \|_{L^6_{\rm per}(\Gamma_L)} \le \widetilde C \quad \mbox{and} \quad \|\nabla W_L\|_{L^2(B_R)} \le \|\nabla W_L\|_{L^2_{\rm per}(\Gamma_L)} \le \widetilde C.
$$
We can therefore extract from $(W_L)_{L \in \N^\ast}$ a subsequence $(W_{L^0_n})_{n \in \N}$ such that $(W_{L^0_n}|_{B_{R_0}})_{n \in \N}$ converges weakly in $H^1(B_{R_0})$, strongly in $L^p(B_{R_0})$ for all $1 \le p < 6$, and almost everywhere in $B_{R_0}$ to some $W^0 \in H^1(B_{R_0})$, for which
$$
\| W^0 \|_{L^6(B_{R_0})} \le  \widetilde C \quad \mbox{and} \quad \|\nabla W^0\|_{L^2(B_{R_0})} \le \widetilde C.
$$
By recursion, we then extract from $(W_{L^k_n})_{n \in \N}$ a subsequence $(W_{L^{k+1}_n})_{n \in \N}$ such that $(W_{L^{k+1}_n}|_{B_{R_{k+1}}})_{n \in \N}$ converges weakly in $H^1(B_{R_{k+1}})$, strongly in $L^p(B_{R_{k+1}})$ for all $1 \le p < 6$, and almost everywhere in $B_{R_{k+1}}$ to some $W^{k+1} \in H^1(B_{R_{k+1}})$, for which
\begin{equation} \label{eq:boundW}
\| W^{k+1} \|_{L^6(B_{R_{k+1}})} \le  \widetilde C \quad \mbox{and} \quad \|\nabla W^{k+1}\|_{L^2(B_{R_{k+1}})} \le \widetilde C.
\end{equation}
Necessarily, $W^{k+1}|_{B_{R_k}} = W^k$. Let $L_n = L^n_n$ and let $W$ be the function of $H^1_{\rm loc}(\R^3)$ defined by $W|_{B_{R_k}}=W^k$ for all $k \in \N$ (this definition is consistent since $W^{k+1}|_{B_{R_k}} = W^k$). The sequence $(W_{L_n})_{n \in \N}$ converges to $W$ weakly in $H^1_{\rm loc}(\R^3)$, strongly in $L^p_{\rm loc}(\R^3)$ for all $1 \le p < 6$ and almost everywhere in $\R^3$. Besides, as (\ref{eq:boundW}) holds for all $k$, we also have
$$
\| W \|_{L^6(R^3)} \le  \widetilde C \quad \mbox{and} \quad \|\nabla W\|_{L^2(\R^3)} \le \widetilde C.
$$
Letting $n$ go to infinity in (\ref{eq:defWL}) with $L=L_n$, we get
$$
- \Delta W = 4\pi \rho.
$$
Introducing the dual
$$
{\mathcal C}' = \left\{ V \in L^6(\R^3) \; | \; \nabla V \in (L^2(\R^3))^3 \right\},
$$
of ${\mathcal C}$, we can reformulate the above results as $W \in {\mathcal C}'$ and $-\Delta W = 4\pi \rho$. As $-\Delta$~is an isomorphism from ${\mathcal C}'$ to ${\mathcal C}$, we necessarily have $\rho \in {\mathcal C}$. From (\ref{eq:boundWL}), we infer that for each $R > 0$, 
$$
\frac{1}{4\pi} \|\nabla W\|_{L^2(B_R)} \le \liminf_{L \to \infty} D_{\cR_L}(\rho_L,\rho_L).
$$
Letting $R$ go to infinity, we end up with (\ref{eq:liminfD}). By uniqueness of the limit, the whole sequence $(W_L)_{L \in \N^\ast}$ converges to $W$ weakly in $H^1_{\rm loc}(\R^3)$, and strongly in $L^p_{\rm loc}(\R^3)$ for all $1 \le p < 6$.

\medskip

Let $p > 6/5$, $(v_L)_{L \in \N}$ be a sequence of functions on $L^p_{\rm loc}(\R^3)$ such that $v_L \in L^p_{\rm per}(\Gamma_L)$ for all $L \in \N^\ast$, and converging to some $v \in L^p_{\rm loc}(\R^3)$ weakly in $L^p_{\rm loc}(\R^3)$, and $\phi \in C^\infty_c(\R^3)$. We have, for $L$ large enough,
\begin{eqnarray*}
D_{\cR_L}(\rho_L,v_L\phi) &=& \int_{\R^3} W_Lv_L\phi - g_1L^{-1} \left( \int_{\Gamma_L} \rho_L \right) \left( \int_{\Gamma_L} v_L\phi \right) \\
&=& \int_{{\rm Supp}(\phi)}(W_L\phi) v_L - g_1 L^{-1} \left( \int_{\Gamma_L} \rho_L \right) \left( \int_{{\rm Supp}(\phi)} v_L\phi \right) \\
& \dps \mathop{\longrightarrow}_{L \to \infty} & \int_{{\rm Supp}(\phi)} W\phi v = D(\rho,v\phi),
\end{eqnarray*}
which proves (\ref{eq:CVD}).
\end{proof}

\medskip

Let us introduce for each $L \in \N^\ast$ the bounded linear operator
\begin{eqnarray} \label{eq:defiL}
i_L \; : \; L^2(\R^3) & \rightarrow & L^2_{\rm per}(\Gamma_L) \\
v & \mapsto & \sum_{z \in {\mathcal R}_L} (\chi_{\Gamma_L}v)(\cdot-z) \nonumber
\end{eqnarray}
and its adjoint $i_L^\ast \in {\mathcal L}(L^2_{\rm per}(\Gamma_L),L^2(\R^3))$. Note that for all $v_L \in L^2_{\rm per}(\Gamma_L)$, $i_L^\ast v_L=\chi_{\Gamma_L}v_L$ and $i_Li_L^\ast=1_{L^2_{\rm per}(\Gamma_L)}$. As $C^\infty_c(\R^3) \subset H^1(\R^3)$, the domain of the self-adjoint operator $(H^0_{\rm per}-\epsilon^0_{\rm F})^{1/2}$, the function $(H^0_{\rm per}-\epsilon^0_{\rm F})^{1/2}\phi$ is in $L^2(\R^3)$. Using the same abuse of notation as above, we can also consider $H^0_{\rm per}$ as a self-adjoint operator on $L^2_{\rm per}(\Gamma_L)$ with domain $H^2_{\rm per}(\Gamma_L)$ and introduce the function $i_L^\ast(H^0_{\rm per}-\epsilon^0_{\rm F})^{1/2}i_L\phi$, which is well-defined in $L^2(\R^3)$. 

\medskip

\begin{lem} \label{lem:weakiL} Let $\phi \in C^\infty_c(\R^3)$. The sequence $(i_L^\ast(H^0_{\rm per}-\epsilon^0_{\rm F})^{1/2}i_L\phi)_{L \in \N^\ast}$ converges to $(H^0_{\rm per}-\epsilon^0_{\rm F})^{1/2}\phi$ in $L^2(\R^3)$.
\end{lem}

\medskip

\begin{proof} According to Bloch-Floquet theory~\cite{ReeSim4}, each $f \in L^2(\R^3)$ can be decomposed as
$$
f(x) = \frac{1}{|\Gamma_1^\ast|} \int_{\Gamma_1^\ast} f_k(x) \, e^{ik\cdot x}\, dk
$$
where $f_k$ is the function of $L^2_{\rm per}(\Gamma_1)$ defined for almost all $k \in \R^3$ by 
$$
f_k(x) = \sum_{R \in {\mathcal R}_1} f(x+R) e^{-ik \cdot (x+R)}. 
$$
Recall that 
$$
\forall (f,g) \in L^2(\R^3) \times L^2(\R^3), \quad (f,g)_{L^2(\R^3)} =  \frac{1}{|\Gamma_1^\ast|} \int_{\Gamma_1^\ast} (f_k,g_k)_{L^2_{\rm per}(\Gamma_1)} \, dk.
$$
The operator $H^0_{\rm per}$, considered as a self-adjoint operator on $L^2(\R^3)$, commutes with the translations of the lattice ${\mathcal R}_1$ and can therefore be decomposed as
$$
H^0_{\rm per} = \frac{1}{|\Gamma_1^\ast|} \int_{\Gamma_1^\ast} (H^0_{\rm per})_k \, dk
$$ 
where $(H^0_{\rm per})_k$ is the self-adjoint operator on $L^2_{\rm per}(\Gamma_1)$ with domain $H^2_{\rm per}(\Gamma_1)$ defined by
$$
(H^0_{\rm per})_k = -\Delta -2 i k \cdot \nabla + |k|^2 + \frac 53 (\rho^0_{\rm per})^{2/3}+V^0_{\rm per}.
$$
Let $\phi$ and $\psi$ be two functions of $C^\infty_c(\R^3)$. Simple calculations show that for $L$ large enough
\begin{equation} \label{eq:limphipsi}
(i_L^\ast(H^0_{\rm per}-\epsilon^0_{\rm F})^{1/2}i_L\phi,\psi)_{L^2(\R^3)} = \sum_{k \in \Gamma_1^\ast \cap {\mathcal R}_L^\ast} L^{-3} ((H^0_{\rm per}-\epsilon^0_{\rm F})_k^{1/2} \phi_k,\psi_k)_{L^2_{\rm per}(\Gamma_1)} , 
\end{equation}
and
\begin{equation} \label{eq:CVnorm}
\|i_L^\ast(H^0_{\rm per}-\epsilon^0_{\rm F})^{1/2}i_L\phi\|_{L^2(\R^3)}^2 
=  \|(H^0_{\rm per}-\epsilon^0_{\rm F})^{1/2}\phi\|_{L^2(\R^3)}^2.
\end{equation}
The sequence $(i_L^\ast(H^0_{\rm per}-\epsilon^0_{\rm F})^{1/2}i_L\phi)_{L \in \N^\ast}$ therefore is bounded in $L^2(\R^3)$, hence possesses a weakly converging subsequence. 

Besides, the function $k \mapsto ((H^0_{\rm per}-\epsilon^0_{\rm F})^{1/2}_k \phi_k,\psi_k)_{L^2_{\rm per}(\Gamma_1)}$ is continuous on $\overline{\Gamma_1^\ast}$ since
$$
((H^0_{\rm per}-\epsilon^0_{\rm F})^{1/2}_k \phi_k,\psi_k)_{L^2_{\rm per}(\Gamma_1)} =
((H^0_{\rm per}-\epsilon^0_{\rm F}+1)_k^{-1}(H^0_{\rm per}-\epsilon^0_{\rm F})^{1/2}_k \phi_k,(H^0_{\rm per}-\epsilon^0_{\rm F}+1)_k\psi_k)_{L^2_{\rm per}(\Gamma_1)} 
$$
with $k \mapsto \phi_k$ and $k \mapsto (H^0_{\rm per}-\epsilon^0_{\rm F}+1)_k\psi_k$ continuous from $\overline{\Gamma_1^\ast}$ to $L^2_{\rm per}(\Gamma_1)$ and $k \mapsto (H^0_{\rm per}-\epsilon^0_{\rm F}+1)_k^{-1}(H^0_{\rm per}-\epsilon^0_{\rm F})^{1/2}_k$ continuous from $\overline{\Gamma_1^\ast}$ to ${\mathcal L}(L^2_{\rm per}(\Gamma_1))$. Interpreting  (\ref{eq:limphipsi}) as a Riemann sum, we obtain 
\begin{eqnarray*}
\lim_{L \to \infty} (i_L^\ast(H^0_{\rm per}-\epsilon^0_{\rm F})^{1/2}i_L\phi,\psi)_{L^2(\R^3)} & = &  ((H^0_{\rm per}-\epsilon^0_{\rm F})^{1/2}\phi,\psi)_{L^2(\R^3)}.
\end{eqnarray*}
The above result allows to identify $(H^0_{\rm per}-\epsilon^0_{\rm F})^{1/2}\phi$ as the weak limit of the sequence $(i_L^\ast(H^0_{\rm per}-\epsilon^0_{\rm F})^{1/2}i_L\phi)_{L \in \N^\ast}$, and (\ref{eq:CVnorm}) shows that the convergence actually holds strongly in $L^2(\R^3)$.
\end{proof}

\subsection{Proof of Proposition~\ref{prop:EU}}

Let $(v_n)_{n \in \N}$ be a minimizing sequence for (\ref{eq:TFWgen}). As each of the three terms of $E_\cR^{\rm TFW}(\rho^{\rm nuc},\cdot)$ is non-negative, the sequence $(v_n)_{n \in \N}$ is clearly bounded in $H^1_{\rm per}(\Gamma)$, hence converges, up to extraction, to some $u \in H^1_{\rm per}(\Gamma)$, weakly in $H^1_{\rm per}(\Gamma)$, strongly in $L^p_{\rm per}(\Gamma)$ for each $1 \le p < 6$ and almost everywhere in $\R^3$. Passing to the liminf in the energy and to the limit in the constraint, we obtain that $u$ satisfies $E_\cR^{\rm TFW}(\rho^{\rm nuc},u) \le I_\cR(\rho^{\rm nuc},Q)$ and $\int_\Gamma u^2=Q$. Therefore, $u$ is a minimizer of~(\ref{eq:TFWgen}). As $|u| \in H^1_{\rm per}(\Gamma)$, $E_\cR^{\rm TFW}(\rho^{\rm nuc},|u|)=E_\cR^{\rm TFW}(\rho^{\rm nuc},u)$ and $\int_\Gamma |u|^2=\int_\Gamma u^2$, $|u|$ also is a minimizer of (\ref{eq:TFWgen}). Up to replacing $u$ with $|u|$, we can therefore assume that $u \ge 0$ in $\R^3$. Clearly, $-u$ also is a minimizer of (\ref{eq:TFWgen}).

\medskip

Working on the Euler equation (\ref{eq:Euler}), we obtain by elementary elliptic regularity arguments \cite{GT} that $u\in H^4_{\rm per}(\Gamma) \hookrightarrow C^{2}(\R^3) \cap L^{\infty}(\R^3)$, and it follows from Harnack's inequality \cite{GT} that $u>0$ in $\mathbb{R}^3$.

\medskip

Lastly, $v_0$ is a minimizer of (\ref{eq:TFWgen}) if and only if $\rho_0=v_0^2$ is a minimizer of 
\begin{equation} \label{eq:minrho}
\inf \left\{ {\mathcal E}^{\rm TFW}_\cR(\rho^{\rm nuc},\rho), \; \rho \in {\mathcal K}_{\cR,Q} \right\},
\end{equation}
where 
$$
{\mathcal E}^{\rm TFW}_\cR(\rho^{\rm nuc},\rho)= C_{\rm W} \int_{\Gamma} |\nabla \sqrt{\rho}|^2 + C_{\rm TF} \int_{\Gamma} \rho^{5/3} + \frac{1}{2}\mathcal{D}_\cR(\rho^{\rm nuc}-\rho, \rho^{\rm nuc}-\rho),
$$
and
$$
{\mathcal K}_{\cR,Q} = \left\{ \rho \ge 0,  \; \sqrt{\rho} \in H^1_{\rm per}(\Gamma), \; \int_\Gamma \rho = Q \right\}.
$$
The functional $\rho \mapsto {\mathcal E}^{\rm TFW}_\cR(\rho^{\rm nuc},\rho)$ being strictly convex on the convex set ${\mathcal K}$, (\ref{eq:minrho}) has a unique minimizer $\rho_0$ and it holds $\rho_0 = u^2 > 0$. Any minimizer $v_0$ of (\ref{eq:TFWgen}) satisfying $v_0^2 = \rho_0 > 0$, the only minimizers of (\ref{eq:TFWgen}) are $u$ and $-u$.

\subsection{Existence of a minimizer to (\ref{eq:minpbInu})}
\label{sec:minpblnu}

The existence of a minimizer to (\ref{eq:minpbInu}) is an obvious consequence of the following lemma.

\medskip

\begin{lem} \label{lem:boundV} It holds 
\begin{equation} \label{eq:coercE}
\exists \beta > 0 \quad \mbox{s.t.} \quad \forall \nu \in {\mathcal C}, \quad \forall v \in {\mathcal Q}_+, \quad \beta \|v\|_{H^1(\R^3)}^2 \le {\mathcal E}^\nu(v),
\end{equation}
\begin{equation} \label{eq:boundD}
\forall \nu \in {\mathcal C}, \quad \forall v \in {\mathcal Q}_+, \quad  \|u^0_{\rm per}v\|_{\mathcal C}^2 \le {\mathcal E}^\nu(v) + \|v^2\|_{\mathcal C}^2 +\|\nu\|_{\mathcal C}^2,
\end{equation}
and for each $\nu \in {\mathcal C}$, the functional ${\mathcal E}^\nu$ is weakly lower semicontiuous in the closed convex subset ${\mathcal Q}_+$ of ${\mathcal Q}$.
\end{lem}

\medskip

Indeed, if $(v_n)_{n \in \N}$ is a minimizing sequence for (\ref{eq:minpbInu}), we infer from (\ref{eq:coercE}) and (\ref{eq:boundD}) that $(v_n)_{n \in \N}$ is bounded in ${\mathcal Q}$. We can therefore extract from $(v_n)_{n \in \N}$ a subsequence $(v_{n_k})_{k \in \N}$ weakly converging in ${\mathcal Q}$ to some $v_\nu \in {\mathcal Q}$. As ${\mathcal Q}_+$ is convex and strongly closed in ${\mathcal Q}$, it is weakly closed in ${\mathcal Q}$. Hence $v_\nu  \in {\mathcal Q}_+$. Besides, ${\mathcal E}^\nu$ being weakly l.s.c. in ${\mathcal Q}_+$, we obtain 
$$
{\mathcal E}^\nu(v_\nu) \le \liminf_{k \to \infty}{\mathcal E}^\nu(v_{n_k}) = I^\nu.
$$ 
Therefore $v_\nu$ is a minimizer of (\ref{eq:minpbInu}).

\medskip

\begin{proof}[Proof of Lemma~\ref{lem:boundV}]
Using (\ref{eq:boundu0per}), (\ref{eq:boundH0per-epsilon0F}), Lemma~\ref{lem:convex}, and the non-negativity of $D$, we obtain that for all $\nu \in {\mathcal C}$ and all $v \in {\mathcal Q}_+$,
\begin{eqnarray*}
{\mathcal E}^\nu(v) & \ge & \frac{2}{3} m^{4/3} \|v\|_{L^2(\R^3)}^2, 
\end{eqnarray*}
and
\begin{eqnarray*}
{\mathcal E}^\nu(v) & \ge & \|\nabla v\|_{L^2(\R^3)}^2 - \left(\frac 53 M^{4/3}+\|V^0_{\rm per}\|_{L^\infty(\R^3)}\right) \|v\|_{L^2(\R^3)}^2.
\end{eqnarray*}
Therefore, there exists some constant $\beta > 0$ such that
$$
\forall \nu \in {\mathcal C}, \quad \forall v \in {\mathcal Q}_+, \quad {\mathcal E}^\nu(v) \ge \beta \|v\|_{H^1(\R^3)}^2.
$$
Besides, for all $\nu \in {\mathcal C}$ and all $v \in {\mathcal Q}_+$,
\begin{eqnarray*}
D(u^0_{\rm per} v,u^0_{\rm per} v)  & \le &  \frac 12 D(2u^0_{\rm per} v+v^2-\nu,2u^0_{\rm per} v+v^2-\nu) + \frac 12 D(v^2-\nu,v^2-\nu) \nonumber \\ & \le &  {\mathcal E}^\nu(v) + D(v^2,v^2)+D(\nu,\nu). 
\end{eqnarray*}
Hence (\ref{eq:boundD}). 

Let $v \in {\mathcal Q}_+$ and $(v_n)_{n \in \N}$ be a sequence of elements of ${\mathcal Q}_+$ weakly converging to $v$ in~${\mathcal Q}$. As $(v_n)_{n \in \N}$ is weakly converging, it is bounded in ${\mathcal Q}$, which means that $(v_n)_{n \in \N}$ and $(u^0_{\rm per}v_n)_{n \in \N}$ are bounded in $H^1(\R^3)$ and ${\mathcal C}$ respectively. We also notice that $(v_n^2)_{n \in \N}$ is bounded in $L^1(\R^3)\cap L^3(\R^3) \hookrightarrow L^{6/5}(\R^3) \hookrightarrow {\mathcal C}$.

Therefore, we can extract from $(v_n)_{n \in \N}$ a subsequence $(v_{n_k})_{ k\in \N}$ such that 
\begin{itemize}
\item $({\mathcal E}^\nu(v_{n_k}))_{k \in \N}$ converges to $I=\liminf_{n \to \infty} {\mathcal E}^\nu(v_n)$ in $\R_+$;
\item $(v_{n_k})_{k\in \N}$ converges to some $\widetilde v \in H^1(\R^3)$ weakly in $H^1(\R^3)$, strongly in $L^p_{\rm loc}(\R^3)$ for all $1 \le p < 6$ and almost everywhere in $\R^3$;
\item  $(u^0_{\rm per}v_{n_k})_{k \in \N}$ weakly converges in ${\mathcal C}$ to some $w \in {\mathcal C}$;
\item $(v_{n_k}^2)_{k \in \N}$ weakly converges in ${\mathcal C}$ to some $z \in {\mathcal C}$.
\end{itemize}
We can rewrite the last two items above as
$$
\forall V \in {\mathcal C}', \quad \int_{\R^3} u^0_{\rm per}v_{n_k} V \mathop{\longrightarrow}_{k \to \infty} \int_{\R^3} w V, \quad \mbox{and} \quad 
 \quad \int_{\R^3} v_{n_k}^2 V \mathop{\longrightarrow}_{k \to \infty} \int_{\R^3} z V.
$$
Together with the strong convergence of $(v_{n_k})_{k \in \N}$ to $\widetilde v$ in $L^2_{\rm loc}(\R^3)$, this leads to $u^0_{\rm per}\widetilde v=w \in {\mathcal C}$ and $z=\widetilde v^2$. This in turn implies that $(v_{n_k})_{k \in \N}$ weakly converges in ${\mathcal Q}$ to $\widetilde v$. Therefore $\widetilde v=v$. Finally, 
$(v_{n_k})_{k \in \N}$ converges to $v$ weakly in $H^1(\R^3)$ and almost everywhere in~$\R^3$ and $(2u^0_{\rm per}v_{n_k}+v_{n_k}^2-\nu)_{k \in \N}$ weakly converges to $2u^0_{\rm per}v+v^2-\nu$ in ${\mathcal C}$.

It follows from (\ref{eq:boundH0per-epsilon0F}) that 
$$
\langle (H^0_{\rm per}-\epsilon_{\rm F}^0) v,v \rangle_{H^{-1}(\R^3),H^1(\R^3)} \le \liminf_{k \to \infty} \langle (H^0_{\rm per}-\epsilon_{\rm F}^0) v_{n_k},v_{n_k}\rangle_{H^{-1}(\R^3),H^1(\R^3)}.
$$
By Fatou's Lemma,
\begin{eqnarray*}
&& \int_{\R^3} \left( |u^0_{\rm per}+v|^{10/3} - |u^0_{\rm per}|^{10/3} - \frac 53 |u^0_{\rm per}|^{4/3}(2u^0_{\rm per}v+v^2) \right) \\
&& \qquad
\le \liminf_{k \to \infty} 
\int_{\R^3} \left( |u^0_{\rm per}+v_{n_k}|^{10/3} - |u^0_{\rm per}|^{10/3} - \frac 53 |u^0_{\rm per}|^{4/3}(2u^0_{\rm per}v_{n_k}+v_{n_k}^2) \right).
\end{eqnarray*}
Lastly, 
$$
D(2u^0_{\rm per}v+v^2-\nu,2u^0_{\rm per}v+v^2-\nu) \le \liminf_{k \to \infty} D(2u^0_{\rm per}v_{n_k}+v_{n_k}^2-\nu,2u^0_{\rm per}v_{n_k}+v_{n_k}^2-\nu).
$$
Consequently, 
$$
{\mathcal E}^\nu(v) \le \liminf_{k \to \infty} {\mathcal E}^\nu(v_{n_k}) = \liminf_{n \to \infty} {\mathcal E}^\nu(v_n),
$$
which proves that ${\mathcal E}^\nu$ is weakly l.s.c. in ${\mathcal Q}_+$.
\end{proof}

\medskip

Clearly, the functional ${\mathcal E}^\nu$ is $C^1$  in ${\mathcal Q}$ and it holds
\begin{eqnarray*}
\forall h \in {\mathcal Q}, \quad \langle {{\mathcal E}^\nu}'(v), h \rangle_{{\mathcal Q}',{\mathcal Q}} &=& 2 \bigg(
\langle (H^0_{\rm per}-\epsilon^0_{\rm F}) v,h \rangle_{H^{-1}(\R^3),H^1(\R^3)} \\
&& + \frac{5}{3} \int_{\R^3} \left( |u^0_{\rm per}+v|^{7/3} - |u^0_{\rm per}|^{7/3} - |u^0_{\rm per}|^{4/3}v \right) h \nonumber \\
&&  + D(2u^0_{\rm per}v + v^2 - \nu,(u^0_{\rm per}+v)h) \bigg).
\end{eqnarray*}
The minimization set ${\mathcal Q}_+$ being convex, $v_\nu$ satisfies the Euler equation
\begin{equation} \label{eq:Eulerconvex}
\forall v \in {\mathcal Q}_+, \quad \langle {{\mathcal E}^\nu}'(v_\nu),(v-v_\nu) \rangle_{{\mathcal Q}',{\mathcal Q}} \ge 0.
\end{equation}
Let $u_\nu=u^0_{\rm per}+v_\nu$ and 
$$
V = V^0_{\rm per}-\epsilon^0_{\rm F} + \frac{5}{3}|u_{\nu}|^{5/3} + (2u^0_{\rm per}v_\nu+v_\nu^2-\nu) \star |\cdot|^{-1}.
$$
The function $u_\nu$ satisfies $u_\nu \in H^1_{\rm loc}(\R^3)$, $u_\nu \ge 0$ in $\R^3$, and 
\begin{eqnarray*}
\forall \phi \in C^\infty_c(\R^3), \quad \int_{\R^3} \nabla u_\nu \cdot \nabla \phi + \int_{\R^3} V u_\nu \phi & = & \frac 12  \langle {{\mathcal E}^\nu}'(v_\nu), \phi \rangle_{{\mathcal Q}',{\mathcal Q}} \\ & = & \frac 12  \langle {{\mathcal E}^\nu}'(v_\nu),(v_\nu+\phi-v_\nu) \rangle_{{\mathcal Q}',{\mathcal Q}}.
\end{eqnarray*}
This implies that for all $\phi \in C^\infty_c(\R^3)$ such that $\phi \ge 0$ in $\R^3$, 
$$
\int_{\R^3} \nabla u_\nu \cdot \nabla \phi + \int_{\R^3} V u_\nu \phi \ge 0,
$$
since $v_\nu+\phi \in {\mathcal Q}_+$. Therefore, $u_\nu$ is a non-negative supersolution of $-\Delta u + V u = 0$, with 
$V \in L^{18/5}_{\rm loc}(\R^3)$. It follows from Harnack's inequality (see Theorem~5.2 of \cite{Trud73}) that either $u_\nu$ is identically equal to zero in $\R^3$, or for each bounded domain $\Omega$ of $\R^3$, there exists $\eta > 0$ such that $v_\nu \geq -u^0_{\rm per}+\eta$ in $\Omega$. As the first case is excluded since $-u_{per}^0 \notin {\mathcal Q}_+$, (\ref{eq:Eulerconvex}) implies ${{\mathcal E}^\nu}'(v_\nu)=0$, which means that $v_\nu$ is a solution in ${\mathcal Q}_+$ to the elliptic equation~(\ref{eq:eulerpb}).

\medskip

Remarking that 
$$
{\mathcal E}^\nu(v_\nu) \le {\mathcal E}^\nu(0)=\frac 12 D(\nu,\nu) = \frac 12 \|\nu\|_{\mathcal C}^2,
$$
and using (\ref{eq:coercE}), (\ref{eq:boundD}) and Lemma~\ref{lem:boundDL}, we finally get the estimate (\ref{eq:boundvnu}).

\subsection{Uniqueness of the minimizer to (\ref{eq:minpbInu})}

Noticing that
$$
{\mathcal Q}_+ = \left\{ v \in H^1(\R^3) \; | \; (u^0_{\rm per}+v)^2-\rho^0_{\rm per} \in {\mathcal C}, \; u^0_{\rm per} +v \ge 0 \right\}, 
$$
we obtain that $v_\star$ is a minimizer to (\ref{eq:minpbInu}) if and only if $\rho_\star=(u^0_{\rm per}+v_\star)^2$ is a minimizer to
\begin{equation} \label{eq:minG}
\inf \left\{ {\mathcal G}(\rho), \; \rho \in {\mathcal K} \right\}
\end{equation}
where
\begin{eqnarray*}
{\mathcal G}(\rho) &=& J(\rho) + \int_{\R^3} \left(\rho^{5/3}-(\rho^0_{\rm per})^{5/3}-\frac 53 (\rho^0_{\rm per})^{2/3} (\rho-\rho^0_{\rm per}) \right) + \frac 12 D(\rho-\rho^0_{\rm per}-\nu,\rho-\rho^0_{\rm per}-\nu),
\end{eqnarray*}
$$
J(\rho) = \langle (H^0_{\rm per}-\epsilon^0_{\rm F})(\sqrt{\rho}-u^0_{\rm per}),(\sqrt{\rho}-u^0_{\rm per})\rangle_{H^1(\R^3),H^{-1}(\R^3)}.
$$
and 
$$
{\mathcal K} = \left\{ \rho \ge 0 \; | \; \sqrt{\rho}-u^0_{\rm per} \in H^1(\R^3), \; \rho-\rho^0_{\rm per} \in {\mathcal C} \right\}.
$$
To see that ${\mathcal K}$ is convex and that ${\mathcal G}$ is strictly convex on ${\mathcal K}$, we first introduce the set
$$
\widetilde{\mathcal K} = \left\{ \rho \ge 0 \; | \; \sqrt{\rho}-u^0_{\rm per} \in H^1(\R^3) \cap {\mathcal E}'(\R^3) \right\}
$$
where ${\mathcal E}'(\R^3)$ denotes the space of the compactly supported distributions, and observe that for all $\rho \in \widetilde{\mathcal K}$,
$$
J(\rho)=
\int_{\R^3} \left( |\nabla \sqrt{\rho}|^2-|\nabla u^0_{\rm per}|^2 + \left(\frac 53 (\rho^0_{\rm per})^{2/3}+ V^0_{\rm per}-\epsilon^0_{\rm F}\right) (\rho-\rho^0_{\rm per}) \right).
$$
Reasoning as in the proof of the convexity of the functional $\rho \mapsto \int_{\R^3} |\nabla \sqrt{\rho}|^2$ on the convex set $\left\{\rho \ge 0 \; |\; \sqrt\rho \in H^1(\R^3)\right\}$ (see e.g. \cite{LiebLoss}), we obtain that $\widetilde{\mathcal K}$ is convex and that ${\mathcal J}$ is convex on $\widetilde{\mathcal K}$. It then follows that ${\mathcal G}$ is strictly convex on $\widetilde{\mathcal K}$. We finally conclude by a density argument.

As ${\mathcal G}$ is strictly convex on the convex set ${\mathcal K}$, (\ref{eq:minG}) has at most one solution. Therefore, $\rho_\nu=(u^0_{\rm per}+v_\nu)^2$ is the unique solution to (\ref{eq:minG}), and $v_\nu$ is the unique solution to (\ref{eq:minpbInu}).

\subsection{Properties of the unique minimizer of (\ref{eq:minpbInu})}

The Euler equation (\ref{eq:eulerpb}) can be rewritten as
\begin{equation} \label{eq:elliptic1}
- \Delta v_\nu + V_\nu u^0_{\rm per} = f + (\nu \star |\cdot|^{-1})u^0_{\rm per},
\end{equation}
where
$$
f = (\epsilon^0_{\rm F}-V^0_{\rm per}) v_\nu -\frac 53 \left( |u^0_{\rm per}+v_\nu|^{7/3} - |u^0_{\rm per}|^{7/3} \right) -  \left((2u^0_{\rm per}v_\nu + v_\nu^2 - \nu)\star |\cdot|^{-1}\right)v_\nu,
$$
and where $V_\nu = (2u^0_{\rm per}v_\nu+v_\nu^2) \star |\cdot|^{-1}$ satisfies 
\begin{equation} \label{eq:elliptic2}
-\Delta V_\nu = 4\pi (2u^0_{\rm per}v_\nu+v_\nu^2).
\end{equation}
We know that $v_\nu \in H^1(\R^3) \hookrightarrow {\mathcal C}'$ and that $V_\nu \in {\mathcal C}'$ since $(2u^0_{\rm per}v_\nu+v_\nu^2) \in {\mathcal C}$. 
Adding up (\ref{eq:elliptic1}) and (\ref{eq:elliptic2}), we obtain that $W_\nu=v_\nu+V_\nu$ is a solution in ${\mathcal C}'$ to
\begin{equation} \label{eq:equationg}
- \Delta W_\nu + u^0_{\rm per} W_\nu = \widetilde f + (\nu \star |\cdot|^{-1})u^0_{\rm per},
\end{equation}
where $\widetilde f = f + (8\pi +1) u^0_{\rm per}v_\nu+4\pi v_\nu^2 \in L^2(\R^3)$. Since $u^0_{\rm per}$ satisfies (\ref{eq:boundu0per}), the elliptic equation 
$$
- \Delta w + u^0_{\rm per} w = \widetilde f
$$
has a unique variational solution in $H^1(\R^3)$, which we denote by $w_\nu$. Clearly $w_\nu \in~H^2(\R^3)$. The function $\widetilde w_\nu=W_\nu-w_\nu \in {\mathcal C}'$ then is solution to 
\begin{equation} \label{eq:equationw}
- \Delta \widetilde w_\nu + u^0_{\rm per} \widetilde w_\nu = (\nu \star |\cdot|^{-1})u^0_{\rm per}.
\end{equation}
Introducing $\widetilde \rho_\nu = -(4\pi)^{-1} \Delta \widetilde w_\nu \in {\mathcal C}$, (\ref{eq:equationw}) also reads
$$
4\pi \frac{\widetilde \rho_\nu}{u^0_{\rm per}} =  (\nu-\widetilde \rho_\nu) \star |\cdot|^{-1}.
$$
Therefore,
$$
4\pi \int_{\R^3} \frac{\widetilde \rho_\nu^2}{u^0_{\rm per}} = D(\nu-\widetilde \rho_\nu,\widetilde \rho_\nu) < \infty,
$$
which proves that $\widetilde \rho_\nu \in L^2(\R^3)$, hence that $(\nu-\widetilde \rho_\nu) \star |\cdot|^{-1} \in L^2(\R^3)$. As
$$
(\nu-(2u^0_{\rm per}v_\nu+v_\nu^2)) \star |\cdot|^{-1} = \nu \star |\cdot|^{-1} - V_\nu =   (\nu-\widetilde \rho_\nu) \star |\cdot|^{-1}+ \widetilde w_\nu - V_\nu =  (\nu-\widetilde \rho_\nu) \star |\cdot|^{-1}+v_\nu - w_\nu,
$$
we obtain
$$
\Phi^0_\nu = (\nu-(2u^0_{\rm per}v_\nu+v_\nu^2)) \star |\cdot|^{-1} \in L^2(\R^3).
$$
Introducing $\rho^0_\nu = \nu-(2u^0_{\rm per}v_\nu + v_\nu^2)$, the above statement reads
$$
\int_{\R^3} \frac{|\widehat{\rho^0_\nu}(k)|^2}{|k|^4} \, dk < \infty.
$$
Therefore,
\begin{eqnarray*}
\frac{1}{|B_r|} \int_{B_r} |\widehat{\rho^0_\nu}(k)| \, dk &\le& 
\frac{1}{|B_r|} \left( \int_{B_r} |k|^4 \, dk \right)^{1/2} \left( \int_{B_r} \frac{|\widehat{\rho^0_\nu}(k)|^2}{|k|^4} \, dk \right)^{1/2} \\
&=& 3 \left(\frac{r}{28\pi}\right)^{1/2}  \left( \int_{B_r} \frac{|\widehat{\rho^0_\nu}(k)|^2}{|k|^4} \, dk \right)^{1/2} \mathop{\longrightarrow}_{r \to 0} 0.
\end{eqnarray*}
Lastly, rewritting (\ref{eq:elliptic1}) as 
$$
-\Delta v_\nu = f + \Phi^0_\nu  u^0_{\rm per},
$$
we conclude that $v_\nu \in H^2(\R^3)$.

\subsection{End of the proof of Theorem~\ref{Th:defect}}

We have proven in the previous two sections that:
\begin{enumerate}
\item (\ref{eq:minpbInu}) has a unique minimizer $v_\nu$;
\item if $(v_n)_{n \in \N}$ is a minimizing sequence for~(\ref{eq:minpbInu}), we can extract from $(v_n)_{n \in \N}$ a subsequence $(v_{n_k})_{k \in \N}$ which 
converges to $v_\nu$, weakly in $H^1(\R^3)$, and strongly in $L^p_{\rm loc}(\R^3)$ for all $1 \le p < 6$, and such that $(u^0_{\rm per}v_{n_k})_{k \in \N}$ converges 
to $u^0_{\rm per}v_\nu$ weakly in ${\mathcal C}$.
\end{enumerate}
By uniqueness of the limit, this implies that any minimizing sequence $(v_n)_{n \in \N}$ for~(\ref{eq:minpbInu}) converges to $v_\nu$, weakly in $H^1(\R^3)$, and strongly 
in $L^p_{\rm loc}(\R^3)$ for all $1 \le p < 6$, and that  $(u^0_{\rm per}v_n)_{n \in \N}$ converges weakly to $u^0_{\rm per}v_\nu$ in ${\mathcal C}$. 
Lastly, the existence of a minimizing sequence for~(\ref{eq:minpbInu}) satisfying (\ref{eq:CCmin}) is a straightforward consequence of Lemma~\ref{lem:density}.

\subsection{Thermodynamic limit  with a charge constraint}

Let $\nu \in L^1(\R^3) \cap L^2(\R^3)$. Clearly, $v_{\nu,q,L}$ is a minimizer to (\ref{eq:minTFWdefL}) if and only if $u^0_{\rm per}+v_{\nu,q,L}$ is a minimizer to (\ref{eq:TFWgen}) with $\cR=\cR_L$, $\rho^{\rm nuc} = \rho^{\rm nuc}_{\rm per} + \nu_L$ and $Q=ZL^3+q$ such that $u^0_{\rm per}+v_{\nu,q,L} \ge 0$ in $\R^3$. It follows from Proposition~\ref{prop:EU} that  (\ref{eq:minTFWdefL}) has a unique minimizer $v_{\nu,q,L}$, which satisfies $v_{\nu,q,L} \in H^4_{\rm per}(\Gamma_L) \hookrightarrow C^2(\R^3) \cap L^\infty(\R^3)$ and $u^0_{\rm per}+v_{\nu,q,L} > 0$ in $\R^3$, and the Euler equation (\ref{eq:eulerpbL}) for some $\mu_{\nu,q,L} \in \R$. 

Let $\alpha = |\Gamma_1|^{-1} \int_{\Gamma_1}u^0_{\rm per}$. For $L$ large enough, $\alpha^2+q/|\Gamma_L| \ge 0$ and the constant function $z_L=-\alpha + \sqrt{\alpha^2+q/|\Gamma_L|}$ satisfies $z_L \ge -u^0_{\rm per}$ everywhere in $\R^3$ and
$$
\int_{\Gamma_L} (2u^0_{\rm per}z_L+z_L^2) = q.
$$
Using Lemma~\ref{lem:convex}, Lemma~\ref{lem:limitDL}, and the fact that $|z_L| \le CL^{-3}$ for some constant $C$ indenpendent of $L$, we obtain 
\begin{eqnarray}
{\mathcal E}^\nu_L(v_{\nu,q,L}) & \le & {\mathcal E}^{\nu}_L(z_L) \nonumber \\
&=& \int_{\Gamma_L}  \left( |u^0_{\rm per}+z_L|^{10/3} - |u^0_{\rm per}|^{10/3} - \frac{10}3 |u^0_{\rm per}|^{7/3}z_L \right)  +\int_{\Gamma_L} (V^0_{\rm per}-\epsilon^0_{\rm F})z_L^2 \nonumber \\ && +\frac 12 D_{\cR_L}
\left( 2u^0_{\rm per}z_L+z_L^2-\nu_L,2u^0_{\rm per}z_L+z_L^2-\nu_L \right) \quad \mathop{\longrightarrow}_{L \to \infty} \quad D(\nu,\nu). \label{eq:bound22}
\end{eqnarray}
Besides, reasoning as in Section~\ref{sec:minpblnu}, we obtain
\begin{equation} \label{eq:bound21}
\forall v_L \in {\mathcal Q}_{+,L}, \quad {\mathcal E}^\nu_L(v_L) \ge \beta \|v_L\|_{H^1_{\rm per}(\Gamma_L)}^2,
\end{equation}
where the constant $\beta > 0$ is the same as in (\ref{eq:coercE}), and
\begin{eqnarray}
\forall v_L \in {\mathcal Q}_{+,L}, \quad 
D_{\cR_L}(u^0_{\rm per} v_L,u^0_{\rm per} v_L)  &\le& {\mathcal E}^\nu_L(v_L) + \frac 12 D_{\cR_L}(v_L^2-\nu_L,v_L^2-\nu_L) \nonumber \\ &\le &
{\mathcal E}^\nu_L(v_L) +  D_{\cR_L}(v_L^2,v_L^2)+  D_{\cR_L}(\nu_L,\nu_L).
\label{eq:bound23}
\end{eqnarray}
We infer from (\ref{eq:bound22}) and (\ref{eq:bound21}) that for each $q \in \R$, there exists $C_q \in \R_+$ such that
\begin{equation} \label{eq:boundvnuqL}
\forall L \in \N^\ast, \quad \|v_{\nu,q,L}\|_{H^1_{\rm per}(\Gamma_L)} \le C_q.
\end{equation}
By a diagonal extraction process similar to the one used in the proof of Lemma~\ref{lem:CVDL}, we can extract from $(v_{\nu,q,L})_{L \in \N^\ast}$ a subsequence $(v_{\nu,q,L_k})_{k \in \N}$ which converges to some $u_\nu \in H^1(\R^3)$, weakly in $H^1_{\rm loc}(\R^3)$, strongly in $L^p_{\rm loc}(\R^3)$ for all $1 \le p < 6$ and almost everywhere in $\R^3$ and such that
$$
\lim_{k \to \infty} {\mathcal E}^\nu_{L_k} (v_{\nu,q,L_k}) = \liminf_{L \to \infty} {\mathcal E}^\nu_{L} (v_{\nu,q,L}).
$$
In particular $u_\nu \ge - u^0_{\rm per}$ almost everywhere in $\R^3$. 

\medskip

Let us now prove that $u^0_{\rm per}u_\nu \in {\mathcal C}$. First, we notice that it follows from (\ref{eq:bound22}), (\ref{eq:bound23}) and Lemma~\ref{lem:boundDL} that there exists a constant $\widetilde C_q$ such that
\begin{equation} \label{eq:boundDuvL}
D_{\cR_L}(u^0_{\rm per} v_{\nu,q,L},u^0_{\rm per} v_{\nu,q,L}) \le \widetilde C_q.
\end{equation}
Besides,
$$
\left|\int_{\Gamma_L} u^0_{\rm per} v_{\nu,q,L}\right| = \left|\frac 12 \left( q - \int_{\Gamma_L} v_{\nu,q,L}^2 \right)\right| \le \frac{1}{2} \left( |q| + C_q^2 \right),
$$
and $(u^0_{\rm per}v_{\nu,q,L_k})_{k \in \N}$ converges to $u^0_{\rm per}u_\nu$ strongly in $L^2_{\rm loc}(\R^3)$, hence in the distributional sense. It therefore follows from Lemma~\ref{lem:CVDL} that $u^0_{\rm per}u_\nu \in {\mathcal C}$. Thus, $u_\nu \in {\mathcal Q}_+$.

\medskip

As (\ref{eq:eulerpbL}) holds in $H^{-1}_{\rm per}(\Gamma_L)$, we can take $u^0_{\rm per}$ as a test function. Using (\ref{eq:uper}), we obtain
\begin{eqnarray*}
\mu_{\nu,q,L} \left( Z L^3 + \int_{\Gamma_L}  v_{\nu,q,L}u^0_{\rm per}\right)& = & \int_{\Gamma_L} \frac{5}{3} \left( |u^0_{\rm per}+v_{\nu,q,L}|^{7/3} - |u^0_{\rm per}|^{7/3} - |u^0_{\rm per}|^{4/3}v_{\nu,q,L} \right) u^0_{\rm per} 
\\ &&  + D_{\cR_L} \left( (2u^0_{\rm per}v_{\nu,q,L} + v_{\nu,q,L}^2 - \nu_L),(u^0_{\rm per}+v_{\nu,q,L})u^0_{\rm per} \right)  .
\end{eqnarray*}
Using (\ref{eq:boundvnuqL}), (\ref{eq:boundDuvL}) and Lemma~\ref{lem:convex}, we obtain
\begin{eqnarray*}
&&  \left|\int_{\Gamma_L} \frac{5}{3} \left( |u^0_{\rm per}+v_{\nu,q,L}|^{7/3} - |u^0_{\rm per}|^{7/3} - |u^0_{\rm per}|^{4/3}v_{\nu,q,L} \right) u^0_{\rm per} \right| \le C'_q L^{3/2}, \\
&& \left| D_{\cR_L} \left( (2u^0_{\rm per}v_{\nu,q,L} + v_{\nu,q,L}^2 - \nu_L),(u^0_{\rm per}+v_{\nu,q,L})u^0_{\rm per} \right)\right| \le C'_q L^{5/2}, \\
&& \left|   \int_{\Gamma_L}  v_{\nu,q,L}u^0_{\rm per}\right| \le  \frac{1}{2} \left( |q| + C_q^2 \right),
\end{eqnarray*}
for some constant $C'_q$ independent of $L$, which allows us to conclude that $(\mu_{\nu,q,L})_{L \in \N^\ast}$ goes to zero when $L$ goes to infinity.

\medskip

Note that using Lemma~\ref{lem:CVDL}, we can pass to the limit in the Euler equation (\ref{eq:eulerpbL}) in the distributional sense, and prove that $u_\nu$ satisfies
\begin{eqnarray}
&& (H^0_{\rm per} - \epsilon^0_F) u_\nu + \frac{5}{3} \left( |u^0_{\rm per}+u_\nu|^{7/3} - |u^0_{\rm per}|^{7/3} - |u^0_{\rm per}|^{4/3}u_\nu \right) \nonumber \\
&& \qquad\qquad\qquad + \left((2u^0_{\rm per}u_\nu + u_\nu^2 - \nu)\star |\cdot|^{-1}\right)(u^0_{\rm per}+u_\nu) = 0. 
\label{eq:eulerpbunu}
\end{eqnarray}

\medskip

We are now going to prove that ${\mathcal E}^\nu(u_\nu) \le {\mathcal E}^\nu(v_\nu)$, which implies that $u_\nu=v_\nu$ and, by uniqueness of the limit, that the whole sequence $(v_{\nu,q,L})_{L \in \N^\ast}$ converges to $v_\nu$ weakly in $H^1_{\rm loc}(\R^3)$, and strongly in $L^p_{\rm loc}(\R^3)$ for all $1 \le p < 6$.

Let $\epsilon > 0$. From Lemma~\ref{lem:density}, there exists $v_{\nu,q}^\epsilon \in {\mathcal Q}_+\cap {\mathcal C}^2_c(\R^3)$ such that 
$$
\int_{\Gamma_L} (2u^0_{\rm per}v^\epsilon_{\nu,q}+(v^\epsilon_{\nu,q})^2) = q
$$ 
and
$$
{\mathcal E}^\nu(v_\nu) \le {\mathcal E}^\nu(v_{\nu,q}^\epsilon) \le {\mathcal E}^\nu(v_{\nu,q})+\epsilon.
$$
For $L$ large enough, the ${\mathcal R}_L$-periodic function $v^\epsilon_{\nu,q,L}$ defined by $v^\epsilon_{\nu,q,L}|_{\Gamma_L}=v^\epsilon_{\nu,q}|_{\Gamma_L}$ is in the minimization set of (\ref{eq:minTFWdefL}). Using Lemma~\ref{lem:limitDL} and the fact that $v^\epsilon_{\nu,q}$ is compactly supported, we have for $L$ large enough $v^\epsilon_{\nu,q,L} \in {\mathcal Q}_{+,L}$ and 
\begin{eqnarray*}
{\mathcal E}^\nu_L(v_{\nu,q,L}) \le  {\mathcal E}^\nu_L(v_{\nu,q,L}^\epsilon)  &=&   \langle (H^0_{\rm per}-\epsilon^0_{\rm F})v_{\nu,q}^\epsilon,v_{\nu,q}^\epsilon\rangle_{H^{-1}(\R^3),H^1(\R^3)}\\ && + \int_{\R^3} \left( |u^0_{\rm per}+v_{\nu,q}^\epsilon|^{10/3} - |u^0_{\rm per}|^{10/3} - \frac 53 |u^0_{\rm per}|^{4/3}(2u^0_{\rm per}v_{\nu,q}^\epsilon+(v_{\nu,q}^\epsilon)^2) \right) \nonumber \\
&& + \frac 12 D_{\cR_L}\left(2u^0_{\rm per}v^\epsilon_{\nu,q,L}+(v^\epsilon_{\nu,q,L})^2-\nu_L,2u^0_{\rm per}v^\epsilon_{\nu,q,L}+(v^\epsilon_{\nu,q,L})^2-\nu_L\right) \\
&& \mathop{\longrightarrow}_{L \to \infty} {\mathcal E}^\nu(v_{\nu,q}^\epsilon). 
\end{eqnarray*}
Therefore, for each $\epsilon > 0$,
$$
{\mathcal E}^\nu_L(v_{\nu,q,L}) \le {\mathcal E}^\nu(v_\nu) +2 \epsilon,
$$
for $L$ large enough, so that
\begin{equation} \label{eq:limsup}
\limsup_{L \to \infty}{\mathcal E}^\nu_L(v_{\nu,q,L}) \le {\mathcal E}^\nu(v_\nu).
\end{equation}
We are now going to prove that
\begin{equation} \label{eq:liminf}
{\mathcal E}^\nu(u_\nu) \le \liminf_{L \to \infty}{\mathcal E}^\nu_{L}(v_{\nu,q,L}). 
\end{equation}
For each $k \in \N$, we denote by 
$$
\widetilde v_k:=i_{L_k}^\ast v_{\nu,q,L_k} \quad \mbox{and} \quad w_k := i_{L_k}^\ast(H^0_{\rm per}-\epsilon^0_{\rm F})^{1/2}v_{\nu,q,L_k},
$$
where the operator $i_{L_k}$ is defined by (\ref{eq:defiL}).
As $\|\widetilde v_k\|_{L^2(\R^3)} = \|v_{\nu,q,L_k}\|_{L^2_{\rm per}(\Gamma_{L_k})}$ and 
$$
\|w_k\|_{L^2(\R^3)}^2 = \langle (H^0_{\rm per}-\epsilon^0_{\rm F})v_{\nu,q,L_k},v_{\nu,q,L_k}\rangle_{H^{-1}_{\rm per}(\Gamma_{L_k}),H^{-1}_{\rm per}(\Gamma_{L_k})},
$$ 
we can extract from $(\widetilde v_k)_{k \in \N}$ and $(w_k)_{k \in \N}$ subsequences $(\widetilde v_{k_n})_{n \in \N}$ and $(w_{k_n})_{n \in \N}$ which weakly converge in $L^2(\R^3)$ to some $\widetilde v \in L^2(\R^3)$ and $w \in L^2(\R^3)$ respectively, and such that 
$$
\lim_{n \to \infty} {\mathcal E}^\nu(v_{\nu,q,L_{k_n}}) = \liminf_{L \to \infty} {\mathcal E}^\nu(v_{\nu,q,L}).
$$
As $(v_{\nu,q,L_k})_{k \in \N}$ converges to $u_\nu$ strongly in $L^2_{\rm loc}(\R^3)$, we have $\widetilde v = u_\nu$. Let us now prove that $w=(H^0_{\rm per}-\epsilon^0_{\rm F})^{1/2}u_\nu$. For each $\phi \in C^\infty_c(\R^3)$, we infer from Lemma~\ref{lem:weakiL} that
\begin{eqnarray*}
(w,\phi)_{L^2(\R^3)}
 &=& \lim_{n \to \infty} (i_{L_{k_n}}^\ast(H^0_{\rm per}-\epsilon^0_{\rm F})^{1/2} v_{\nu,q,L_{k_n}}, \phi)_{L^2(\R^3)}
\\ &=& \lim_{n \to \infty} (i_{L_{k_n}}^\ast(H^0_{\rm per}-\epsilon^0_{\rm F})^{1/2}i_{L_{k_n}} \widetilde v_{k_n}, \phi)_{L^2(\R^3)}
\\ &=& \lim_{n \to \infty} (\widetilde v_{k_n},i_{L_{k_n}}^\ast (H^0_{\rm per}-\epsilon^0_{\rm F})^{1/2} i_{L_{k_n}} \phi)_{L^2(\R^3)}\\ &=& ( u_\nu ,(H^0_{\rm per}-\epsilon^0_{\rm F})^{1/2} \phi)_{L^2(\R^3)} = ((H^0_{\rm per}-\epsilon^0_{\rm F})^{1/2} u_\nu , \phi)_{L^2(\R^3)}. 
\end{eqnarray*}
As a consequence, $w=(H^0_{\rm per}-\epsilon^0_{\rm F})^{1/2} u_\nu$.

\medskip

Using the weak convergence of $w_{k_n}$ to $w=(H^0_{\rm per}-\epsilon^0_{\rm F})^{1/2}u_\nu$, Fatou's Lemma and Lemma~\ref{lem:CVDL}, we thus obtain 
\begin{eqnarray*}
{\mathcal E}^\nu(u_\nu) & = & \|(H^0_{\rm per}-\epsilon^0_{\rm F})^{1/2}u_\nu\|_{L^2(\R^3)}^2 \\
&& + \int_{\R^3} \left( |u^0_{\rm per}+u_\nu|^{10/3} - |u^0_{\rm per}|^{10/3} - \frac 53 |u^0_{\rm per}|^{4/3}(2u^0_{\rm per}u_\nu+u_\nu^2) \right) \nonumber \\
&& + \frac 12 D\left(2u^0_{\rm per}u_\nu+u_\nu^2-\nu,2u^0_{\rm per}u_\nu+u_\nu^2-\nu\right) \\ &\le& \liminf_{n \to \infty}{\mathcal E}^\nu(v_{\nu,q,L_{k_n}}) =  \liminf_{L \to \infty}{\mathcal E}^\nu(v_{\nu,q,L}).
\end{eqnarray*}
Hence (\ref{eq:liminf}). Gathering (\ref{eq:limsup}) and (\ref{eq:liminf}), we obtain that ${\mathcal E}^\nu(u_\nu)\le {\mathcal E}^\nu(v_\nu)$ and therefore that $u_\nu=v_\nu$ since $u_\nu \in {\mathcal Q}_+$ and (\ref{eq:minpbInu}) has a unique minimizer.

\subsection{Thermodynamic limit  without a charge constraint}

Let $(v_n)_{n \in \N}$ be a minimizing sequence for (\ref{eq:minTFWdefLf}). For all $\eta>0$, for $n$ large enough, 
$$
\beta \|v_n\|_{H^1_{\rm per}(\Gamma_L)}^2 \le {\mathcal E}^\nu_L(v_n) \le {\mathcal E}^\nu_L(0) + \eta = \frac 12 D_{\cR_L}(\nu_L,\nu_L)+ \eta.
$$
Thus, $(v_n)_{n \in \N}$ is bounded in $H^1_{\rm per}(\Gamma_L)$. Extracting a converging subsequence and passing to the liminf in the energy, we obtain a solution $v_{\nu,L}$ to (\ref{eq:minTFWdefLf}), such that
\begin{equation} \label{eq:boundvL}
\beta \|v_{\nu,L}\|_{H^1_{\rm per}(\Gamma_L)}^2 \le \frac 12 D_{\cR_L}(\nu_L,\nu_L).
\end{equation}
We also get
\begin{equation} \label{eq:boundDuvLD}
D_{\cR_L}(u^0_{\rm per} v_{\nu,L},u^0_{\rm per} v_{\nu,L}) \le \widetilde C,
\end{equation}
for some constant $\widetilde C$ independent of $L$.

Clearly, $u^0_{\rm per}+v_{\nu,L}$ is a non-negative solution to
$$
\inf \left\{ E^{\rm TFW}_{\cR_L}(\rho^{\rm nuc}_{\rm per}+\nu_L,w_L), \; w_L \in H^1_{\rm per}(\Gamma_L) \right\}.
$$
Reasoning as in the proof of Proposition~\ref{prop:EU}, we obtain that $u^0_{\rm per}+v_{\nu,L}$ is the only non-negative solution to the above problem, and therefore that $v_{\nu,L}$ is the unique solution to~(\ref{eq:minTFWdefLf}). Besides, $v_{\nu,L} \in H^4_{\rm per}(\Gamma_L)$, $u^0_{\rm per}+v_{\nu,L} > 0$ in $\R^3$, and $v_{\nu,L}$ is solution to the Euler equation~(\ref{eq:eulerpbL2}), which holds in $H^{-1}_{\rm per}(\Gamma_L)$. Taking $u^0_{\rm per}$ as a test function, we get
\begin{eqnarray*}
&& \!\!\!\!\!\!\!\!\!\!
\int_{\Gamma_L} \frac{5}{3} \left( |u^0_{\rm per}+v_{\nu,L}|^{7/3} - |u^0_{\rm per}|^{7/3} - |u^0_{\rm per}|^{4/3}v_{\nu,L} \right) u^0_{\rm per} 
\\ &&  + D_{\cR_L} \left( (2u^0_{\rm per}v_{\nu,L} + v_{\nu,L}^2 - \nu_L),v_{\nu,L}u^0_{\rm per} \right)+ D_{\cR_L} \left( (2u^0_{\rm per}v_{\nu,L} + v_{\nu,L}^2 - \nu_L),(u^0_{\rm per})^2 \right)  =0 .
\end{eqnarray*}
We now remark that the third term can be rewritten as
\begin{eqnarray} 
D_{\cR_L} \left( (2u^0_{\rm per}v_{\nu,L} + v_{\nu,L}^2 - \nu_L),(u^0_{\rm per})^2 \right) &=& g_1 ZL^2 \left( \int_{\Gamma_L} (2u^0_{\rm per}v_{\nu,L} + v_{\nu,L}^2 - \nu_L) \right) \nonumber \\
&& + 
\int_{\Gamma_L} (2u^0_{\rm per}v_{\nu,L} + v_{\nu,L}^2 - \nu_L) W^0_{\rm per},
\label{eq:D01}
\end{eqnarray}
where, as above, $g_1=|\Gamma_1|^{-1} \int_{\Gamma_1}G_{\cR_1}$ and where $W^0_{\rm per}$ is the unique solution in $H^2_{\rm per}(\Gamma_1)$ to 
$$
\left\{ \begin{array}{l} \dps -\Delta W^0_{\rm per} = 4\pi \left( \rho^0_{\rm per} - |\Gamma_1|^{-1} Z \right) \\
\dps W^0_{\rm per} \mbox{ $\cR_1$-periodic}, \quad \int_{\Gamma_1} W^0_{\rm per} = 0 .
\end{array} \right.
$$
We finally obtain
\begin{eqnarray*}
g_1 ZL^2 \left( \int_{\Gamma_L} (\nu-(2u^0_{\rm per}v_{\nu,L} + v_{\nu,L}^2)) \right) & = & \int_{\Gamma_L} \frac{5}{3} \left( |u^0_{\rm per}+v_{\nu,L}|^{7/3} - |u^0_{\rm per}|^{7/3} - |u^0_{\rm per}|^{4/3}v_{\nu,L} \right) u^0_{\rm per} 
\\ &&  + D_{\cR_L} \left( (2u^0_{\rm per}v_{\nu,L} + v_{\nu,L}^2 - \nu_L),v_{\nu,L}u^0_{\rm per} \right) \\ && + \int_{\Gamma_L} (2u^0_{\rm per}v_{\nu,L} + v_{\nu,L}^2 - \nu_L) W^0_{\rm per}.
\end{eqnarray*}
As the right hand side is bounded by $CL^{3/2}$ for a constant $C$ independent of $L$, it holds
$$
\lim_{L \to \infty}  \int_{\Gamma_L} (\nu-(2u^0_{\rm per}v_{\nu,L} + v_{\nu,L}^2)) = 0.
$$

\medskip

Proceeding {\it mutatis mutandis} as in the previous section, it can be shown that the sequence $(v_{\nu,L})_{L \in \N^\ast}$ converges weakly in $H^1_{\rm loc}(\R^3)$ and strongly in $L^p_{\rm loc}(\R^3)$ for all $1\le p < 6$, towards the unique solution $v_\nu$ to (\ref{eq:minpbInu}).

\end{document}